\documentclass[lettersize,journal]{IEEEtran}

\usepackage{amsmath,amsfonts}
\usepackage{algorithm}
\usepackage{array}
\usepackage{textcomp}
\usepackage{stfloats}
\usepackage{url}
\usepackage{verbatim}
\usepackage{graphicx}
\usepackage{cite}
\usepackage{xspace}
\usepackage{xcolor}
\usepackage{enumitem}
\usepackage{subfigure}
\usepackage{makecell}
\usepackage{booktabs}
\usepackage[noend]{algpseudocode}
\usepackage{amsthm}
\usepackage{amssymb}

\newcommand{\name}{TelApart\xspace}
\newcommand{\sys}{TelApart\xspace}
\newcommand{\anonisp}{ISP-X\xspace}
\newcommand{\Pnmdata}{PNM data\xspace}
\newcommand{\pnmdata}{PNM data\xspace}

\newcommand{\pnmmetrics}{performance metrics\xspace}
\newcommand{\proactive}{batch mode\xspace}

\newcommand{\reactive}{reactive mode\xspace}

\newcommand{\FN}{fNode\xspace}
\newcommand{\FNs}{fNodes\xspace}

\newcommand{\orangedots}{dots\xspace}
\newcommand{\reddot}{square\xspace}
\newcommand{\reddots}{squares\xspace}

\newcommand{\greendots}{triangles\xspace}

\newcommand{\TODO}[1]{\textcolor[HTML]{e41a1c}{{(#1)}}}

\DeclareMathOperator*{\argmin}{arg\,min}

\newtheorem{theorem}{Theorem}

\newcommand{\eg}{{e.g.}\xspace}
\newcommand{\ie}{{i.e.}\xspace}
\newcommand{\ea}{{et al.}\xspace}

\newcommand{\jy}[1]{\textcolor{blue}{[JY: #1]}}
\newcommand{\zzy}[1]{#1\xspace}

\newcommand{\zcut}[1]{}

\expandafter\def\expandafter\normalsize\expandafter{%
    \normalsize%
    \setlength\abovedisplayskip{0pt}%
    \setlength\belowdisplayskip{0pt}%
    \setlength\abovedisplayshortskip{-8pt}%
    \setlength\belowdisplayshortskip{2pt}%
}

\begin{document}

\title{\name: Differentiating Network Faults from Customer-Premise Faults in Cable Broadband Networks}

\author{Jiyao Hu, Zhenyu Zhou, Xiaowei Yang,~\IEEEmembership{Duke University}
}



\maketitle

\begin{abstract}

Two types of radio frequency (RF) impairments frequently occur in a cable broadband network:  impairments that 
occur inside a cable network and impairments occur at the edge of the broadband network, i.e., in a
subscriber's premise. Differentiating these two types of faults is
important, as different faults require different types of technical
personnel to repair them. Presently, the cable industry lacks publicly
available tools to automatically diagnose the type of fault. In this work, we present \sys, a fault diagnosis system for cable
broadband networks. \sys uses telemetry data collected by the
Proactive Network Maintenance (PNM) infrastructure in cable networks
to effectively differentiate the type of fault.  Integral to \sys's
design is an unsupervised machine learning model that groups cable
devices sharing similar anomalous patterns together. We use metrics
derived from an ISP's customer trouble tickets to programmatically
tune the model's hyper-parameters so that an ISP can deploy \sys in
various conditions without hand-tuning its hyper-parameters. We also
address the data challenge that the telemetry data collected by the PNM
system contain numerous missing,
duplicated, and unaligned data points.  Using real-world data
contributed by a cable ISP, we show that \sys can effectively identify
different types of faults. 
\end{abstract}

\section{Introduction}\label{sec:intro}

%



In September 2020, CNN featured a story about a village in
Wales~\cite{bbc2020news}. For a period of 18 months, the broadband
cable Internet of every household in the village mysteriously crashed
every morning. Technicians repeatedly visited the village and even
replaced cables in the area to no avail. Finally, a team of outside
experts visited the village. After laborious testing, they caught the
culprit: an old TV that turned on every day at the news hour. This
incident signifies the challenges of troubleshooting last-mile
networks.  Cable broadband networks have a hybrid fiber-coaxial (HFC)
architecture. The coaxial segments of a cable network consist of many
components, including amplifiers, cable connectors, and cable
shieldings. These components are exposed to real-world conditions such
as inclement weather, radio frequency (RF) interference, backhoeing,
and wild-animal spoliation. Consequently, failures frequently occur in
those networks at haphazard locations, leading to time-consuming and
error-prone fault diagnosis.



There are two types of common faults that impact the service quality
and availability of cable broadband networks. The first type of fault
is a maintenance issue, where a faulty component lies inside the customer-shared 
network infrastructure.  The second type of fault is a service issue, where a
faulty component lies in a subscriber's premise.  It is important to
distinguish a maintenance issue from a service issue because repairing
each type of fault requires a different type of technician. If a cable
ISP makes a wrong diagnosis, they may send a service technician to a
subscriber's home for a maintenance issue or vice versa. In such
cases, the technician is unable to repair the fault, resulting in a waste of 
operational resources and a delay in failure repair time.

%





Presently, 
there does not exist a publicly available tool in the cable
industry that automatically differentiates maintenance issues from
service issues, 
and we are not aware of any automated private tools for separating maintenance issues from service issues.
In a typical scenario, when a cable ISP receives a customer
call, if the customer service representative cannot resolve the issue
over the phone, the ISP will first dispatch a service technician to
the customer's home by default. If the service technician cannot fix the issue,
and more customers in the nearby area start to report similar issues,
the issue is then escalated to the maintenance team for further
investigation. Incorrect diagnoses and unnecessary dispatches are
commonplace in the operations of cable broadband networks. 
They significantly contribute to an ISP's operational costs.

In this work, we aim to develop a publicly available turn-key solution
to help cable ISPs distinguish maintenance issues from service issues.
The cable Internet standard--Data Over Cable Service Interface
Specifications (DOCSIS)--has a built-in Proactive Network Maintenance
(PNM) system that periodically collects performance metrics from cable
devices~\cite{cablelabs2016docsis3}. Ideally, we could solve the fault
diagnosis problem by training a machine learning classifier with labeled \pnmdata. However, this straightforward approach is complicated by a few
practical challenges. First, there does not exist high-quality labeled \pnmdata in the public domain. PNM data are proprietary and
ISPs do not share them with the general public.  In addition, the operating conditions of
different cable networks vary significantly, and even for the same
network, the operating conditions change over time. Therefore, to use a supervised machine learning model for fault diagnosis, an ISP must
continuously obtain labeled \pnmdata for training and for coping with
model drift by its own staff. This overhead could offset a key advantage of an automated fault diagnosis tool. Furthermore, network operators prefer simple and explainable models to advanced machine learning algorithms such as neural networks, whose classification rules are difficult to comprehend~\cite{piet2023ggfast}. Second,
a typical machine learning model includes hyper-parameters, and each
ISP must tune those hyper-parameters to optimize its performance for
their networks. The need for hand-tuning again will reduce the value
of a tool aiming at improving the efficiency of cable ISP operations.
Finally, the PNM infrastructure's data collection process is
unreliable in nature, resulting in missing, misaligned, and duplicated
data points. Currently, there are no publicly available machine
  learning tools that can accurately differentiate maintenance issues
  from service issues using these data.



This paper presents the design, implementation, and evaluation of
\sys, a system that aims to effectively separate maintenance issues from
service issues without the need for labeled data and hand-tuning
hyper-parameters. In addition, it also works with unreliably-collected
\pnmdata. We make three essential design decisions to overcome the
practical challenges faced by \sys's design.  First, we decouple fault
detection from fault diagnosis. We first apply unsupervised learning
(clustering) to group cable devices that share similar \pnmdata
patterns together and employ a separate fault detection module to
detect which clusters experience a common anomaly. The size of a
cluster separates a maintenance issue from an individual service
issue. Second, we use metrics derived from customer tickets
and apply optimization techniques to programmatically tune the
hyper-parameters of \sys's machine learning model. As a result, \sys
can effectively identify groups of devices affected by shared
infrastructure faults without labeled training data or manually tuning hyper-parameters. Finally, we develop data pre-processing techniques to
convert raw \pnmdata into the input format suitable for \sys's unsupervised machine
learning model.


To evaluate \sys's design, we collaborate with a U.S. regional ISP and
obtain their \pnmdata and customer tickets from more than 70k
cable modems in a span of 14 months. For the purpose of evaluation, we
manually labeled a small set of devices as healthy, experiencing a
maintenance issue, or experiencing a service issue.  Using the
manually labeled data as ground truth, we show that \sys's clustering
algorithm achieves a rand index~\cite{rand1971objective} of 0.91 (1.0
being the highest), indicating that \sys's fault diagnosis is highly
accurate.  Each customer ticket we obtain includes a field that
describes an operator's diagnosis of the issue as being maintenance or
service. Using \sys's diagnosis results as ground truth, we
estimate that $38.52\%$ of the dispatches employed
an incorrect type of technician and could have been avoided if \sys had been
deployed. Furthermore, we compare a few metrics derived from
customer tickets during the time periods where \sys concludes
there are maintenance issues with those \sys concludes as service
issues. An example of such metrics is the time elapsed between a
customer reporting a ticket and a fault occurring.  We observe
significant statistical differences between the two sets of metrics,
further validating that \sys can effectively differentiate maintenance
issues from service issues. Field tests from our collaborating ISP confirmed \sys's effectiveness ``at the task of classifying defects as service or maintenance~\cite{anonisp2023private}.''

To the best of our knowledge, \sys is the first system in the public domain that can
effectively separate maintenance issues from service
issues in cable broadband networks. Its design incorporates
the following key features:

\begin{enumerate} 
[topsep=0.1ex,leftmargin=2\labelsep,wide=0pt]
\setlength{\itemsep}{2pt}
\setlength{\parskip}{2pt}
\item \sys employs an architecture that uses unsupervised learning and
  anomaly detection applied to PNM data to accurately distinguish
  maintenance issues from service issues without the need for labeled training data. 
\item We develop data pre-processing techniques that enable
  machine learning systems to use \pnmdata as inputs, which include
  missing, duplicated, and misaligned data points.
\item By utilizing customer ticket statistics and optimization
  techniques, \sys automates hyper-parameter tuning, enabling any cable ISPs to deploy the system in their networks' operating environments
  without the need for hand-tuning hyper-parameters.
  \end{enumerate}

\section{Background}
In this section, we introduce the cable broadband network architecture, describe the datasets we obtain, and present relevant efforts in this area. 

\subsection{Hybrid Fiber-Coaxial (HFC) Architecture}

\begin{figure}[tb!]
\centering
    \includegraphics[width=.4\textwidth]{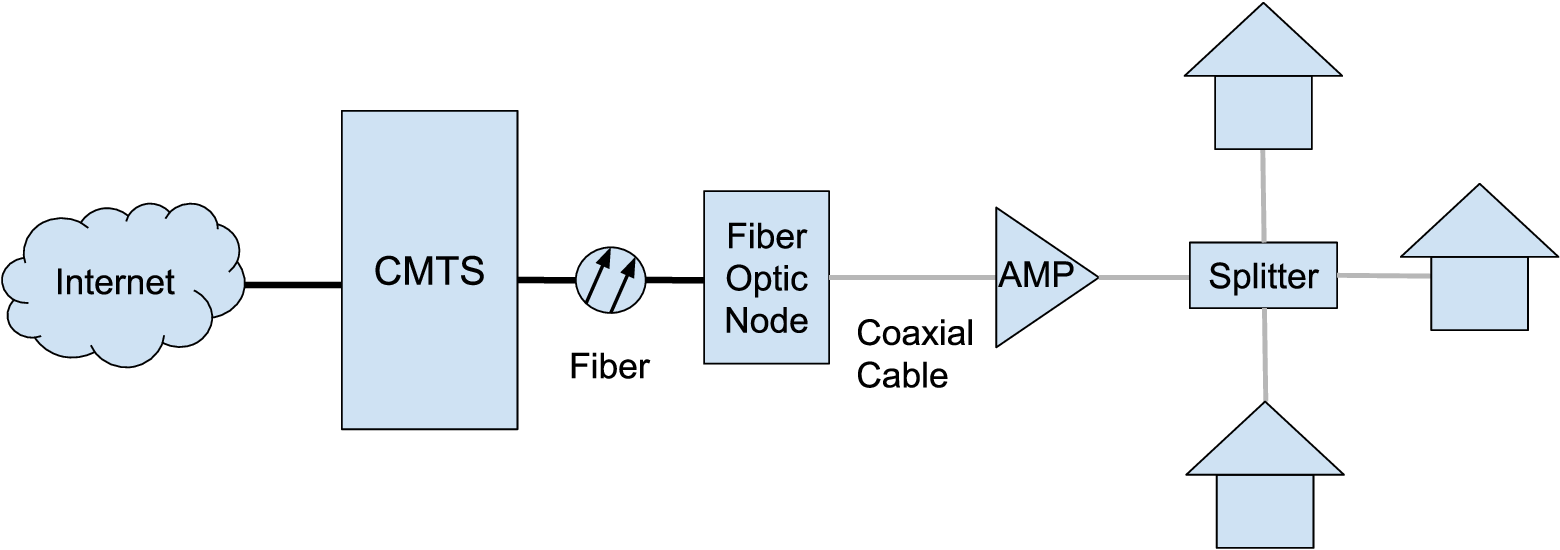}
    \caption{\small\textbf{An overview of the Hybrid Fiber Coaxial (HFC) architecture.}}
    \label{f:hfc}
\end{figure}


A cable network has a hybrid fiber and coaxial (HFC) architecture.
Figure~\ref{f:hfc} presents a schematic illustration of this
architecture. The Cable Modem Termination System (CMTS) is located at the
headend of an ISP, providing Internet connections to cable modems
located at subscribers' premises. A CMTS connects to devices called
Fiber Optic Nodes (\FNs) via optical fibers. \FNs
convert radio frequency (RF) signals to light signals and vice
versa. They connect to individual homes or businesses via coaxial
cables. An \FN typically serves a few hundred 
customers. For example, in our dataset, the average number of
customers an \FN serves is around 250.  To alleviate RF signal
attenuation, RF amplifiers are deployed in the coaxial segments of an
HFC network to ensure that  RF signals delivered to end users are
strong and of high quality.


Fault diagnosis remains a challenging issue in cable broadband
networks. The coaxial segment of a cable network is prone to
radio frequency impairments.  In response to this challenge, the cable
industry developed the PNM network monitoring framework to facilitate
anomaly detection and fault diagnosis. A monitoring server sends
periodic Simple Network Management Protocol (SNMP)~\cite{stallings1993snmp} queries to collect
performance metrics from customer cable devices as well as a CMTS. We
refer to such data as \pnmdata.



\subsection{Datasets}\label{sec:dataset}

For this study, an anonymous cable ISP (\anonisp),
provided us with the \pnmdata collected from their networks and their
customer tickets. We describe each dataset in turn.

\paragraph{\Pnmdata} 
These data are collected from 70k+ customer devices and span a period
of 14 months. A monitoring server collects the data approximately
every 4 hours in-band. According to \anonisp, their
infrastructure cannot support a shorter collection interval.
PNM infrastructure~\cite{ron2020full} can collect both upstream (from a customer device
  to the CMTS) and downstream (the reverse direction) telemetry
  data in DOCSIS 3.0 and 3.1 cable devices. At the time of this study, \anonisp did not automatically collect downstream PNM data (partly because their daily operations did not depend much on PNM data prior to this study and collecting downstream PNM data is more resource-consuming as there are more downstream channels than upstream channels).  Therefore, this study is based on upstream PNM data. 


%

Each data collection point includes the several fields relevant to this study:

\begin{itemize} [topsep=0.1ex,leftmargin=2\labelsep,wide=0pt]
        \setlength\itemsep{0em}
\setlength{\itemsep}{0pt}
\setlength{\parskip}{0pt}
\item \emph{Timestamp:} the time at which the data point was received by the collection server.
\item \emph{Anonymized device id:} the hashed MAC address of a
  customer device.
\item \emph{\FN}: the identifier of the fiber optic node serving this device.
\item \emph{SNR}: the  signal-to-noise ratio of a customer device's transmission signal
  measured at a CMTS.
\item \emph{Tx Power}: the signal transmission power when a cable device
  sends a signal. This signal is recorded by a customer's device and collected by the collection server from each cable device. 
\item \emph{Rx Power}: the power of the received signal at the CMTS.
\item \emph{Pre-Equalization Coefficients}: the coefficients used by the pre-equalizer component in a customer device to compensate for linear signal distortions in coaxial cables. 
\end{itemize}


\paragraph{Customer tickets} \anonisp creates a customer
ticket to document how it handles a customer call. Each ticket
contains several fields, including the customer's account number,
the ticket open time, the ticket close time (if any), a short
description of the problem and its actions, and a category of the
issue based on the ISP's diagnosis. The category includes two
classes: a part-of-primary ticket or not. The last field is crucial to
this work. A part-of-primary ticket indicates that the ISP considers
the issues the customers are experiencing a maintenance issue. Thus,
it groups the tickets as one conceptual ``primary'' ticket. All
part-of-primary tickets that belong to the same maintenance issue have
the same primary ticket identifier. In this work, we refer to
part-of-primary tickets as \emph{maintenance tickets} and other
infrastructure-related tickets as \emph{service tickets}.

    The customer ticket data we obtain span over the same period as the
    \pnmdata and are from customers located within the same
    networks.

 %






\subsection{Ethical Considerations} The PNM data we received includes encrypted cable devices' MAC addresses and  scrambled location data (latitude + $noise_1$ and longitude + $noise_2$).
\zcut{For each of the cable devices, we also have their scrambled location data (latitude + $noise_1$ and longitude + $noise_2$.) All researchers involved in this research have concurred not to deanonymize the MAC addresses and location information. }
Other PNM data are related to the physical signal properties of the cable devices. 
We discussed the data and the scope of this research with the IRB of our organization before conducting this work. The IRB determined that this work does not meet the definition of research with human
subjects and it is appropriate for us to conduct this study. This work raises no other ethical concerns. 

%


\subsection{Related Work}
\label{sec:cablemon}

\paragraph{PNM Best Practice} The PNM best practice
document~\cite{cablelabs2016docsis3} proposes to use a clustering
algorithm to separate maintenance issues from service issues using the
pre-equalization coefficients collected from cable devices.
Pre-equalization coefficients are frequency-domain data and capture
signal distortions in the frequency domain. Each cluster produced by
the algorithm corresponds to a group of cable devices sharing a
similar signal distortion at one \pnmdata point. However, these
distortions may have already been compensated for and do not manifest
themselves as user-perceivable performance issues. In addition,
pre-equalization coefficients at each data point only reflect the
instantaneous signal distortions and are prone to noise fluctuations.
We find them ineffective in detecting network performance issues
(\S~\ref{sec:compare_to_pnm}). Therefore, this work does not use them as features for
either anomaly detection or fault classification.

Similarly, Volpe \ea~\cite{volpe2021machine} propose to use the full
band spectrum (both upstream and downstream) data and apply
DBSCAN~\cite{ester1996density} to group cable devices sharing the same
anomalous spectrum patterns to reduce operational wastage caused by
erroneous dispatches. The authors acknowledge in their work that there
are no good approaches to tune the hyper-parameters of the system to
set an anomaly detection threshold and the work does not include the
experimental evaluation of the effectiveness of the proposed approach.


Different from previous approaches, \name treats \pnmdata as
time-series data and applies clustering techniques using the time-domain
similarity of \pnmdata.  It tackles the challenges associated with
time-domain data, such as data incompleteness and alignment problems
commonly encountered when data are collected unreliably in production
systems.

\paragraph{CableMon}


\zcut{Similar to this work, }
CableMon~\cite{hu2020cablemon} treats \pnmdata
as time-series data. Differently, it focuses on the task of
determining an anomaly detection threshold for a PNM metric associated
with a cable device. 
\zcut{A novel technique developed by CableMon is to
employ the statistics of customer tickets to set the anomaly detection
thresholds. Although customer tickets are generally noisy and cannot
be used directly as the ground truth for training or evaluation, their
statistics still indicate the trend toward network
faults. Specifically,} CableMon defines \emph{ticketing rate}, \ie, the
average number of customer tickets created in a unit of time, as the
statistics guiding its anomaly detection. Intuitively, the ticketing
rate measures how frequently customer tickets are reported for a
certain device. The authors of CableMon observed that the higher the
ticketing rate is, the more likely there is a fault. \zcut{Given
that more customer tickets also naturally imply more operational
resources for an ISP, CableMon can effectively detect network faults
and reduce operational costs.}

%


\sys is inspired by this per-device anomaly detection work, with the
observation of a gap in distinguishing network fault types. Although
CableMon can tell if a single device has ongoing anomalies, it still
cannot help ISPs determine the best team (maintenance or service) to
dispatch. 
\sys can not only
assist with the decision of whether there is an anomaly, but also
assist with the type of dispatch. 
\zcut{\sys adopts and extends the
ticketing rate metric and the anomaly threshold setting technique
proposed by CableMon.}

\paragraph{Other Related Work}

Orthogonal to \sys, previous study~\cite{hu2022characterizing} has modeled characteristics of cable network faults and showcased physical-layer transmission errors are significant. There exists other fault detection work in the domain of cable
broadband networks
~\cite{lartey2017proactive,rupe2019kickstarting,cablelabs2016docsis3,keller2019proactive,zhou2020proactive,ferreira2020convolutional}. However,
this body of work either uses static anomaly threshold settings or
requires manual labeling of \pnmdata.
There also exist tools that aim
to assist ISPs' manual troubleshooting by offering visualization and
suggestions to
operators~\cite{milley2019proactive,rupe2020profile,volpe2021machine,walsh2009pathtrak}. \sys
aims to automate fault diagnosis without manually labeled data and is
orthogonal to this work. It can be enhanced with the visualization
tools. Likewise, it is possible to design a fault diagnosis system for
cable networks with manually labeled \pnmdata and more advanced
machine learning
techniques~\cite{goodfellow2016deep}. We
chose to explore the design without labeled \pnmdata and evaluate the
hypothesis of whether such a design can be effective.

\zcut{Researchers and practitioners have used customer tickets for network
troubleshooting. NetSieve~\cite{potharaju2013juggling} analyzes
natural language text in customer tickets to automate root cause
analysis. LOTUS~\cite{venkataraman2018assessing} analyzes natural
language text in customer tickets to estimate the impact of network
outages. \sys uses ticketing rates instead of the natural language
text included in customer tickets for setting the hyper-parameters of
its clustering model and for determining anomaly thresholds. It is
possible that adopting natural language analysis can improve its fault
diagnosis performance, but we leave such a study as future work and
opt to use explainable machine learning algorithms and metrics for
operational simplicity.}



\section{Design Rationale} 
\label{sec:overview}


\begin{figure*}[t!]
\centering \subfigure[]{
  \includegraphics[width=.4\textwidth]{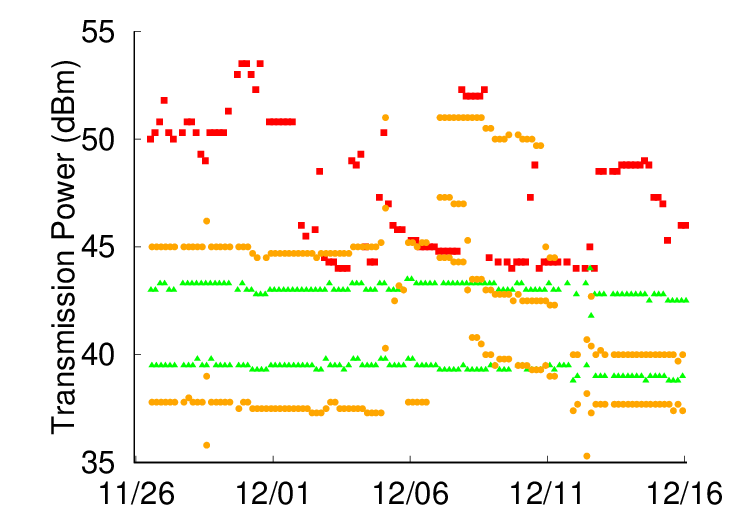}
    \label{f:tx_power_examples}
    } \hspace{10pt}
    \subfigure[]{
    \includegraphics[width=.4\textwidth]{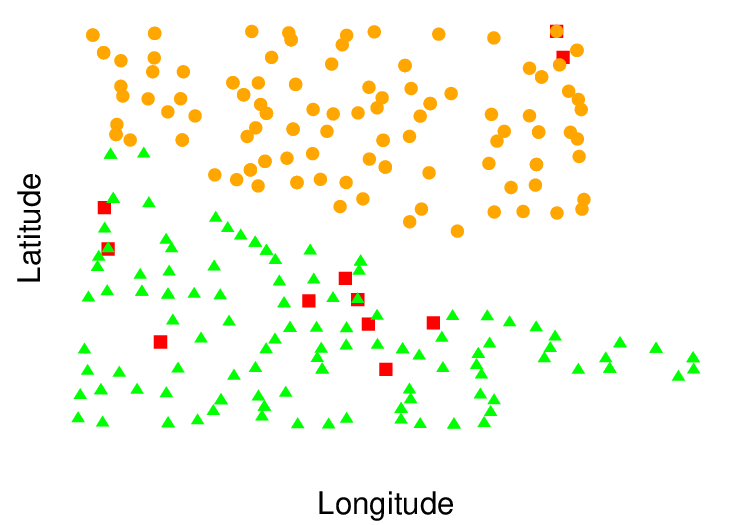}
    \label{f:tx_power_examples_map}
    }
    \vspace{-10pt}
    \caption{{\small{\bf Figure~(a) shows how the transmission powers of \zzy{several} 
    cable devices in the same fiber optical node fluctuate over time. Orange \orangedots are devices that show the same anomalous transmission power patterns. Green \greendots are devices that show normal patterns. Red \reddots are devices that show distinct anomalous patterns. Figure~(b) shows the locations of the cable devices using the same colored icons.}}\label{f:align}}
\end{figure*}


%


We aim to develop an easy-to-deploy fault diagnosis system that can
automatically distinguish maintenance issues from service issues in
cable broadband networks. At first glance, this appears to be a
straightforward task: all we need is a machine learning classifier
that automatically classifies \pnmdata collected from each cable
device as healthy, with a maintenance issue, or with a service issue.
Nevertheless, this simple problem is complicated by several practical
challenges that need to be addressed.

\subsection{Challenges}
\label{sec:challenges}

The initial challenge we face is the high cost and time involved in
acquiring high-quality labeled \pnmdata. Those data are proprietary
and it requires specialized expertise to identify anomalies
accurately.  A domain expert must examine a large number of cable
devices over a sufficiently long period of time to cover all possible
anomalous patterns. As an example to illustrate how tedious this task
is, we depict the transmission powers of six cable devices in
Figure~\ref{f:tx_power_examples}. A domain expert must label thousands
of such figures at a minimum to generate useful training data. 

One might be attempted to utilize customer tickets as automatic
labels: when there is a maintenance ticket, we label the data point as
maintenance, and vice versa. In the absence of a ticket, we label a
data point as healthy. However, our attempt to adopt this approach
revealed its ineffectiveness. This is because customer tickets are
highly prone to noise: customers may call when there are no
infrastructure problems and may not call when there are
problems. Moreover, according to \anonisp, their diagnosis of a
maintenance or service issue is inaccurate. Therefore, this
ticket-based-labeling approach does not lead to high-quality labels.


The second challenge we face is that a machine learning model unavoidably incorporates many
hyper-parameters. However, cable ISPs operate in vastly different
conditions. For instance, each may use different frequency channels or
operate in varying climatic conditions. Therefore, each ISP must label its
own \pnmdata to train and tune the model for effective fault diagnosis
and to mitigate model drift. Yet many cable ISPs do not have dedicated
personnel for such tasks.

Third, unsupervised learning models compare the similarity of their
input data to find similar patterns. However, \pnmdata are collected
unreliably and contain many missing and duplicated data
points. Furthermore, the data collection system introduces randomness
to avoid network congestion. Thus, the data collected from each device
are not temporally aligned. It is preferable to compare data points that are collected close in time to produce meaningful comparisons. However, using the raw \pnmdata as inputs may misalign data points that are distant in data collection times, thereby rendering similarity comparisons inaccurate.


\subsection{Goals}

To mitigate the practical challenges mentioned above, we ask the
  following question: \emph{is it feasible to develop a fault diagnosis
  system for cable networks using unlabeled PNM data and
  hyperparameters that are automatically tuned?} To answer this
  question, we identify the following design requirements for \sys:

%
\paragraph{No manual data labeling:} \sys must not rely on labeled
 data for model training or tuning.

\paragraph{Automated hyper-parameter tuning:} \sys must be able to
tune its hyper-parameters programmatically without human intervention.

\paragraph{Effective despite unreliably-collected data:} \sys must effectively
separate maintenance issues from service issues using the existing
\pnmdata that include missing, duplicated, and unaligned data across
cable devices.

\subsection{Motivation}

To gain insight into how to separate a service issue from a
maintenance issue, we manually examined several anomaly patterns by
plotting PNM metrics described in \S~\ref{sec:dataset}.
Figure~\ref{f:tx_power_examples} shows an example. In this figure, we
sampled the transmission power levels of devices with three anomaly
patterns from an \FN. The orange \orangedots show the
transmission power levels of three devices that exhibit similar
anomalous patterns in the changes of their transmission powers. When a
noise leaks inside a cable transmission channel, a device increases
its data transmission power to overwhelm the noise. So a sudden
increase in transmission power is an indicator of noise invasion. The
green \greendots show the transmission power levels of two devices that are
not impacted by the noise. The red \reddots show the transmission power
levels of a device that exhibits a different anomalous pattern.

Figure~\ref{f:tx_power_examples_map} shows the geographic distribution
of the devices in the \FN. We use the same color coding schemes to
plot the devices. The orange \orangedots plot the scrambled
geographic locations of the devices that exhibit similar anomalous
patterns as depicted in Figure~\ref{f:tx_power_examples}. Each red
\reddot shows the scrambled geographic location of a device that
exhibits a distinct anomalous pattern. And the green \greendots show
the locations of the devices that do not exhibit any anomalous
pattern.

From this data visualization step, we gained the conceptual
understanding that we could use clustering to distinguish a
maintenance issue from a distinct service issue.  In addition, we
  observe that fault detection is independent of clustering, as both
  the healthy devices (the green group) and the unhealthy ones (the
  orange group) form distinct clusters.  These observations, together, motivate \sys's
architecture, which we describe next.

\section{Design}
\label{sec:design}

In this section, we describe \sys's design\zzy{, including the programmatically approach to set hyper-parameters}.

\subsection{System Overview}

\begin{figure}[t!]
\centering
    \includegraphics[width=.4\textwidth]{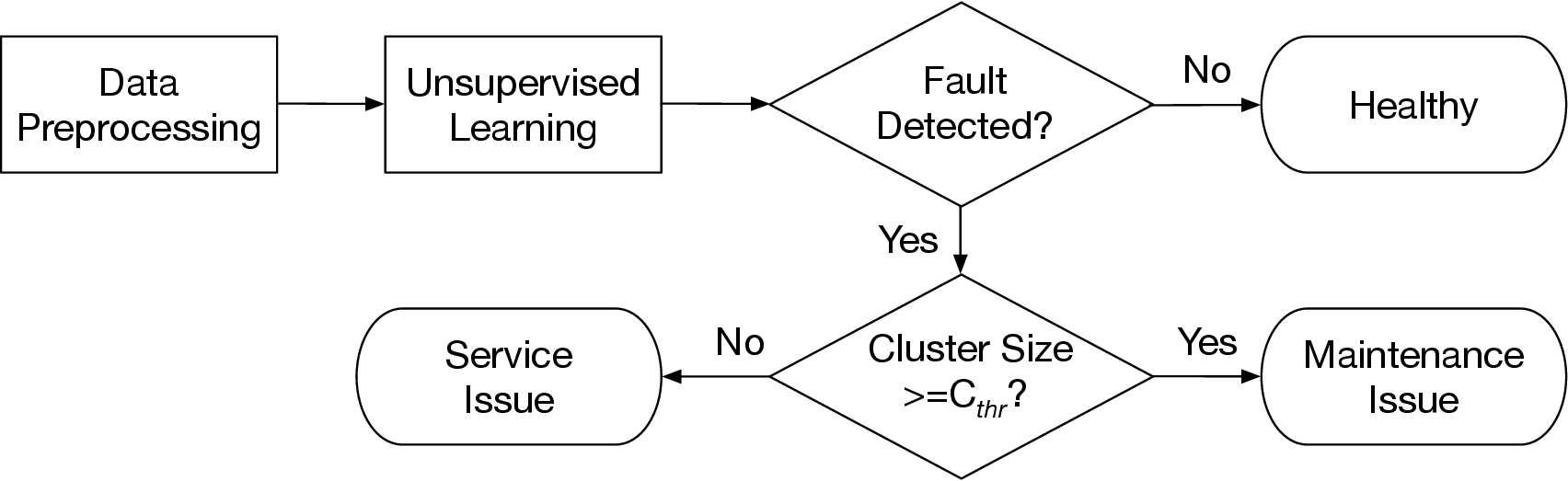}
    \caption{\small\textbf{An overview of \sys's design. $C_{thr}$ is a threshold specified by a network operator based on its network topology.}}
    \label{f:design_overview}
\end{figure}

 %

We make several design choices to address the practical
challenges \sys's design faces. First, we decouple common \pnmdata pattern recognition
from fault detection, as the former can be accomplished with
unsupervised learning (\S~\ref{sec:clustering}).

Second, we design optimization techniques that use customer tickets
and operational knowledge as clues to tune hyperparameters of an
unsupervised learning
model and set the anomaly detection thresholds
programmatically. While customer tickets are noisy and cannot
be used for precisely labeling \pnmdata, they still offer valuable
information: when there is a customer ticket, a network is more likely
to experience an infrastructure problem than when there is no customer
ticket. This observation motivates the design of using the ticketing
rates of maintenance tickets to choose the values of \sys's
hyper-parameters (\S~\ref{sec:set_similarity_threshold}).

Finally, we design techniques to pre-process \pnmdata. After
pre-processing, each device will have the same number of data points
so we can feed the pre-processed data to a clustering model
(\S~\ref{sec:pre-processing}).

Figure~\ref{f:design_overview} shows the workflow of \sys. We first
pre-process the selected raw PNM features from all cable devices in a
fiber optical node. After pre-processing, all devices have the same
number of data points and these data points are aligned in time. We
then feed the pre-processed data to a clustering model, which outputs
clusters of devices that share similar \pnmdata patterns. For each
cluster, we detect whether it exhibits any anomalous pattern. A cluster is identified as ``healthy'' if there are no anomalous patterns. Otherwise, we distinguish a maintenance issue and a service issue with the cluster size: if an anomalous cluster has more than $C_{thr}$ devices, we classify it as experiencing a maintenance issue. Otherwise, the cluster is marked as a service issue. 
\zzy{The threshold $C_{thr}$ is a configuration parameter, chosen as 5 by \anonisp according to their physical networks in this work, \ie, the majority of amplifiers serve more than 5 devices.}

Each invocation of \sys takes the \pnmdata from cable devices
connected to one \FN, as those devices share the same
RF domain. 

\sys's clustering component takes pre-processed \pnmdata as input and
outputs groups of cable devices that share similar \pnmdata
patterns. We select a subset of independent PNM metrics as features
and cluster cable devices by each feature. We describe feature
selection in \S~\ref{sec:feature}.  Each input feature vector includes the
pre-processed \pnmdata points collected between a time interval $(t-d, t]$, where $t$ is the time of diagnosis and $d$ configures
  a look-back duration, \eg, one or two
  days.  \zzy{In practice, \anonisp collects \pnmdata every 4 hours, and all devices have 3 upstream channels. Because \anonisp intends to perform daily detection, we set $d = 1$ day. In this case, we obtain 6 data points in each channel, a total of 18 when
there is no data loss. We also varied $d$ from 1 to 7 days and did
not observe any significant performance difference, indicating a longer look-back window does not benefit the daily detection.}
  
  \sys's
  pre-processing unit aligns each PNM feature vector in time so we can
  compare the similarity of any two devices' feature vectors. We
  define a similarity comparison function for each PNM feature and
  describe them in \S~\ref{sec:feature}.

\subsection{Clustering}
\label{sec:clustering}

We consider and evaluate several clustering
algorithms~\cite{ester1996density,yu2015hierarchical,tan2016introduction,
  timm2002applied} and find that the average-linkage hierarchical
clustering algorithm~\cite{yu2015hierarchical} performs the best
(\S~\ref{sec:compare_choice}).  At a high level, this clustering
algorithm works as follows. For each feature $f$ we selected, the
clustering algorithm aims to group devices with similar feature
vectors together until the similarity between groups of devices falls
below a threshold $s_f$. Specifically, it first treats each device
(described by a feature vector) as a single cluster. Second, it
calculates the similarity between every pair of clusters and finds two
clusters with the highest similarity value. The similarity between two
clusters is calculated by averaging all similarity values between
pairs of devices in the two clusters. Third, the algorithm merges the
two clusters with the highest similarity value into a single
cluster. Next, the algorithm repeats the second and third steps until
only one cluster is left or the highest similarity value between any
two clusters is less than the similarity threshold $s_{f}$. Finally, the
algorithm outputs the clusters that have not been merged.


\subsection{Setting the Similarity Threshold}
\label{sec:set_similarity_threshold}



\begin{figure}[t!]
\centering
    \includegraphics[width=.4\textwidth]{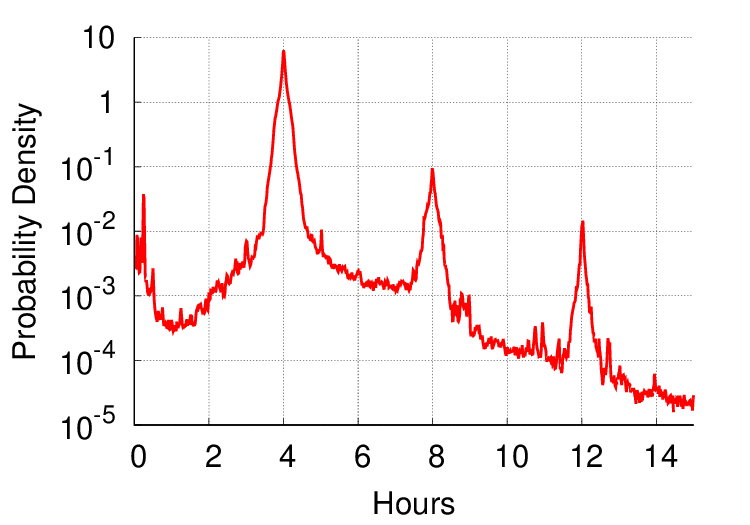}
    \vspace{-5pt}
    \caption{\small\textbf{This figure shows the probability density function of data collection intervals.}}
    \label{f:missing_interval}
\end{figure}

The similarity threshold $s_f$ for each feature $f$ is an important
hyper-parameter and \sys's performance is sensitive to its value.  If
we set the threshold too high, \sys may separate devices that are
affected by the same network fault into multiple clusters. Conversely,
if we set the threshold too low, it may group devices that are
affected by different maintenance issues into the same cluster.

How do we choose a proper similarity threshold? If we had labeled
training data, we could use the grid-search
method~\cite{lavalle2004relationship} to iterate over possible values and set the threshold that
minimizes clustering errors. 
\zcut{That is, we search through the range of
the similarity metric $s_f$ at small steps, \eg, from 0 to 1 with an
increment such as 0.001, and for each value, we compute the false
positives and false negatives using the labeled data and choosing a
value that minimizes the errors.} Lacking of labeled data, we instead use customer ticket statistics to
guide the search for the similarity threshold. Our insight is that if
\sys correctly identifies groups of devices that are impacted by the
same maintenance issue, then on average, we should observe a higher
fraction of maintenance tickets reported by these groups of devices
than other devices. In contrast, if \sys partitions the cable devices
rather randomly, then we should not observe significant statistical
differences of the reported maintenance tickets among different
groups.

Motivated by this insight, we naturally evolve the \emph{ticketing rate} of CableMon (\S~\ref{sec:cablemon}) to \emph{maintenance ticketing rate} and devise the following mechanism to set
the similarity threshold $s_f$ for each PNM feature we use. We
partition the PNM dataset we have into a training set and a testing
set. For each data point $i$ in the training set and for each possible
value of $s_f$, we use \sys to diagnose whether a device $j$ is
impaired by an infrastructure fault and the type of fault, as shown in
Figure~\ref{f:design_overview}. If \sys considers a device experiencing a
maintenance issue, we mark this collection period of this device as a
maintenance event. We use $I_{i,j}$ to denote the length of the data
collection interval between data points $i$ and $i-1$ of device $j$. Similarly, if
\sys considers a device experiencing a service issue, we mark the
collection period of the device as a service event. We then count the
number of maintenance tickets reported by all devices
during all collection periods that are marked as maintenance issues
and compute a maintenance ticketing rate during maintenance events as
\begin{equation}
  R_{m,M} = \frac{K_{m,M}} {\sum_{i,j} I^{M}_{i,j}}
  \label{eq:rt_m}
\end{equation}
where $K$ denotes the number of tickets, $R$ denotes the ticketing
rate, the first subscript $m$ denotes maintenance tickets, and the
second subscript $M$ denotes a diagnosed maintenance issue, and
  $I^{M}_{i,j}$ is the length of a collection period that is marked as
 experiencing a  maintenance issue.

We also count the number of maintenance tickets $K_{m,S}$ reported by
all devices during all collection periods that are marked as service
issues.  We compute a maintenance ticketing rate during service events
as
\begin{equation}
  R_{m,S} = \frac{K_{m,S}}{\sum_{i,j} I^{S}_{i,j}}
  \label{eq:rt_s}
\end{equation}
where $S$ indicates a diagnosed service issue, and
$I^{S}_{i,j}$ is the length of a collection period marked as
experiencing a service issue.  We define the maintenance Ticketing
Rate Ratio ($TRR_m$) as
\begin{equation}
  TRR_m = \frac{R_{m,M}}{R_{m,S}}
  \label{eq:trrm}
\end{equation}
For each feature $f$ \sys uses, we use grid-search to find the
similarity threshold value $s_f$ that maximizes $TRR_m$.
In Appendix~\ref{sec:cluster_proof}, we prove that the $s_f$ maximizing $TRR_m$
yields the
optimal clustering result: it minimizes both false positives (\ie, a device without any maintenance issues is detected
as with a maintenance issue) and false negatives (\ie, a device with a maintenance issue is
detected as without any maintenance issues). 
Intuitively, based on \anonisp's fault diagnosis
process, maintenance tickets contain fewer false positives than
service tickets. Thus, if we assume the operator-labeled
maintenance tickets approximate the unknown but existing ground truth
of maintenance events and \sys's fault detection mechanism is
accurate, then the maintenance ticketing rate during maintenance
events approximates \sys's true positives and the maintenance ticketing
rate during service events approximates \sys's false negatives. Maximizing
the ratio of the two ticketing rates leads to high true positives
and  low false negatives.

\begin{figure}[!t]
\centering
  \includegraphics[width=0.6\linewidth]{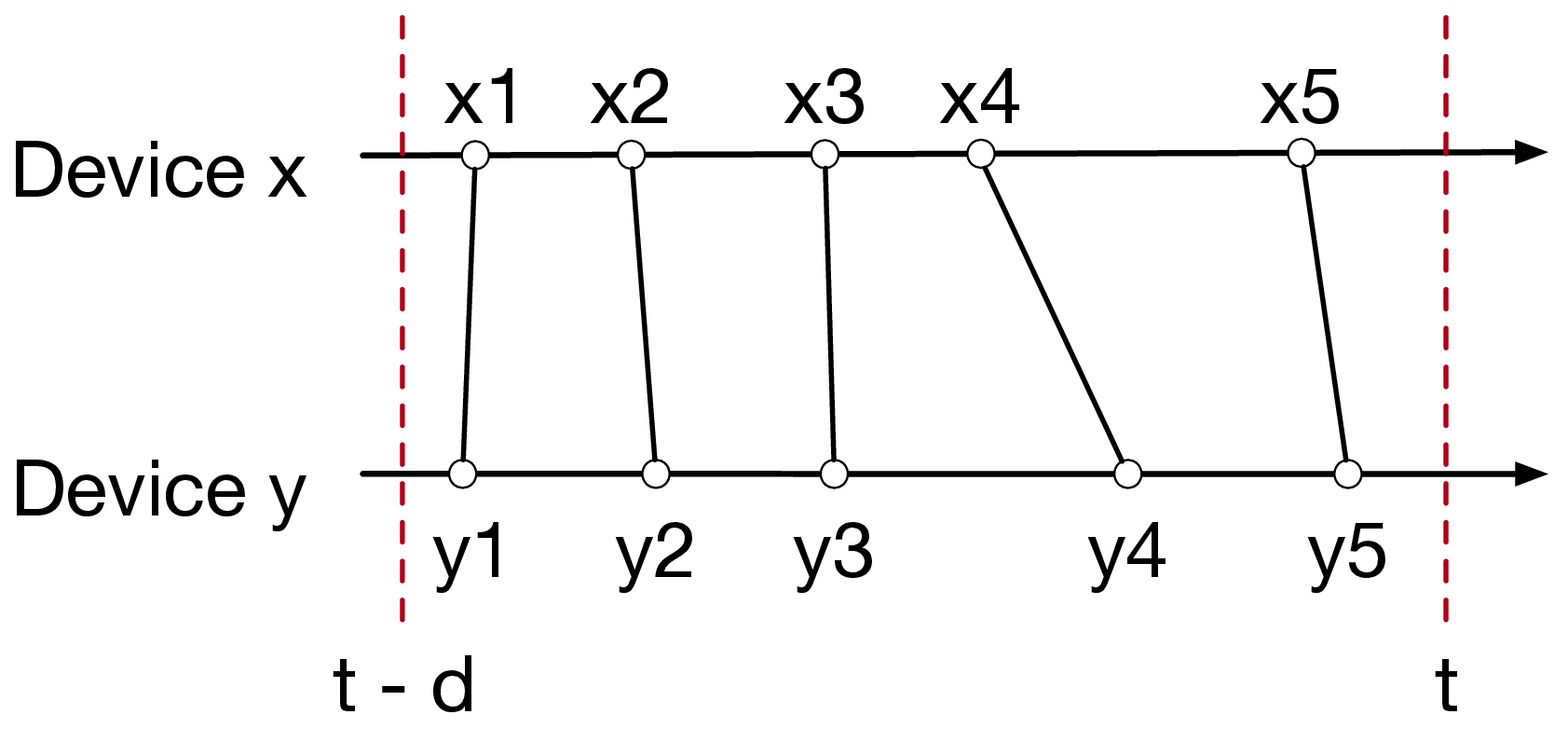}
  \caption{\small\textbf{This figure shows an example that without inferring missed data points, 
a naive alignment algorithm may pair up two data points $x_4$ and $y_4$, which are not close in time.}}\label{f:align_f}

\end{figure}

\subsection{Feature Selection and Comparison}
\label{sec:feature}

The PNM system collects many metrics, but for simplicity and
computation efficiency, it is desirable to use only a minimum set of
effective and independent features for fault diagnosis. We describe
how we choose the features \sys uses and the similarity function we
choose to compare each feature.

We start with a candidate set of features that CableMon~\cite{hu2020cablemon} finds effective in detecting infrastructure
faults in cable broadband networks.  Some of these features contain
the instantaneous values measured at the data collection times, while
others are cumulative values (e.g. codeword error counters) over
time. We find that the instantaneous metrics, including SNR, Tx power,
and Rx power, are effective features for grouping devices with shared
maintenance issues together - they achieve high maintenance
ticketing rate ratios when used as clustering features. On the other hand,
cumulative metrics, although effective in detecting
anomalies~\cite{hu2020cablemon}, are not effective as clustering
features.  For example, the values of codeword
error counters are affected by whether users actively use the Internet
or not. Devices that share the same maintenance issue may
or may not have highly correlated codeword error counters if the
subscribers' usage patterns differ. Hence, this work uses only
instantaneous metrics: SNR, Tx power, and Rx power. Among them, Tx and
Rx powers are statistically correlated. Finally, we retain two independent
features: SNR and Tx power for clustering. 

We use the Pearson correlation coefficient~\cite{benesty2009pearson}
as the similarity metric to compare two devices' SNR and Tx power
values. We choose the Pearson correlation coefficient because it
  measures the linear correlation between two vectors. When a
  maintenance issue occurs, the impacted devices' PNM metrics will
  increase or decrease simultaneously to adapt to the changed network
  conditions.  Pearson coefficient can capture the synchronized
  changes well regardless of the absolute values of two devices' PNM
  metrics, as the baseline values of different devices' Tx
  powers could differ significantly.

A new feature we uncover in this work is a feature vector that encodes
the missed data collection points. \sys's data pre-processing module
can infer which data points each device misses.  Many factors can
cause data missing, \eg, the loss of a \pnmdata collection request or
response, a faulty device, or a network outage.  Intuitively, if a
group of devices simultaneously miss data collection points, it could
indicate a maintenance issue. Led by this insight, we generate a
binary feature vector with 0 indicating a missing data collection
point.  We find that using this feature vector (referred to as missing
hereafter) as a clustering feature leads to a high maintenance
ticketing rate ratio. So we also include missing as a clustering
feature.  We use one minus the normalized hamming distance as the
similarity metric between two devices' missing features. Normalized
hamming distance is a metric that computes the distance between two
binary strings~\cite{pappalardo2009hamfast} and is suitable for
comparing the similarity of two binary feature vectors.

\subsection{Fault Detection}
\label{sec:anomaly-detection}

After \sys's clustering module identifies groups of devices sharing
similar \pnmdata patterns, it invokes a fault detection module on each
cluster (which could contain a single device), as shown in
Figure~\ref{f:design_overview}. When the module detects an anomaly
  for any device in a cluster, it flags the entire cluster as
  anomalous.

Fault detection is relatively independent of \sys's data
pre-processing and clustering modules. \sys adopts the
state-of-the-art fault detection techniques proposed
in CableMon~\cite{hu2020cablemon}, as they can detect infrastructure faults
using \pnmdata without labeled training data nor static
thresholds.  
CableMon~\cite{hu2020cablemon} detects network faults with dynamic thresholds on \pnmmetrics of a device's \pnmdata.
A device is detected with an anomaly if there are \pnmmetrics below or beyond the corresponding threshold.
CableMon elaborately selects these thresholds maximizing the ticketing rate (\S~\ref{sec:cablemon}) such that the detected anomalies are associated with the most customer tickets, which implies larger impacts on user experience and more operational costs.  \zzy{\sys treats these fault detection thresholds as hyper-parameters auto-tuned by CableMon.}
However, the techniques themselves cannot differentiate
the types of faults. 
\sys can incorporate any fault detection module
that meets its design goals listed in \S~\ref{sec:overview}.


%
%

\subsection{Data Pre-Processing}
\label{sec:pre-processing}


%

To compare the similarity of two devices' \pnmdata, it is desirable to
compare data points collected closest in time together.  However, the
PNM infrastructure uses SNMP to collect data, which is unreliable. We
observe many missed and duplicated data points in our dataset.  In
addition, in the same data collection epoch, different devices respond
at different times to avoid congestion.  So the timestamps of
different devices' data points could span a wide range even in the
same data collection epoch.

\begin{figure*}[t!]
\centering \subfigure[]{
  \includegraphics[width=.3\textwidth]{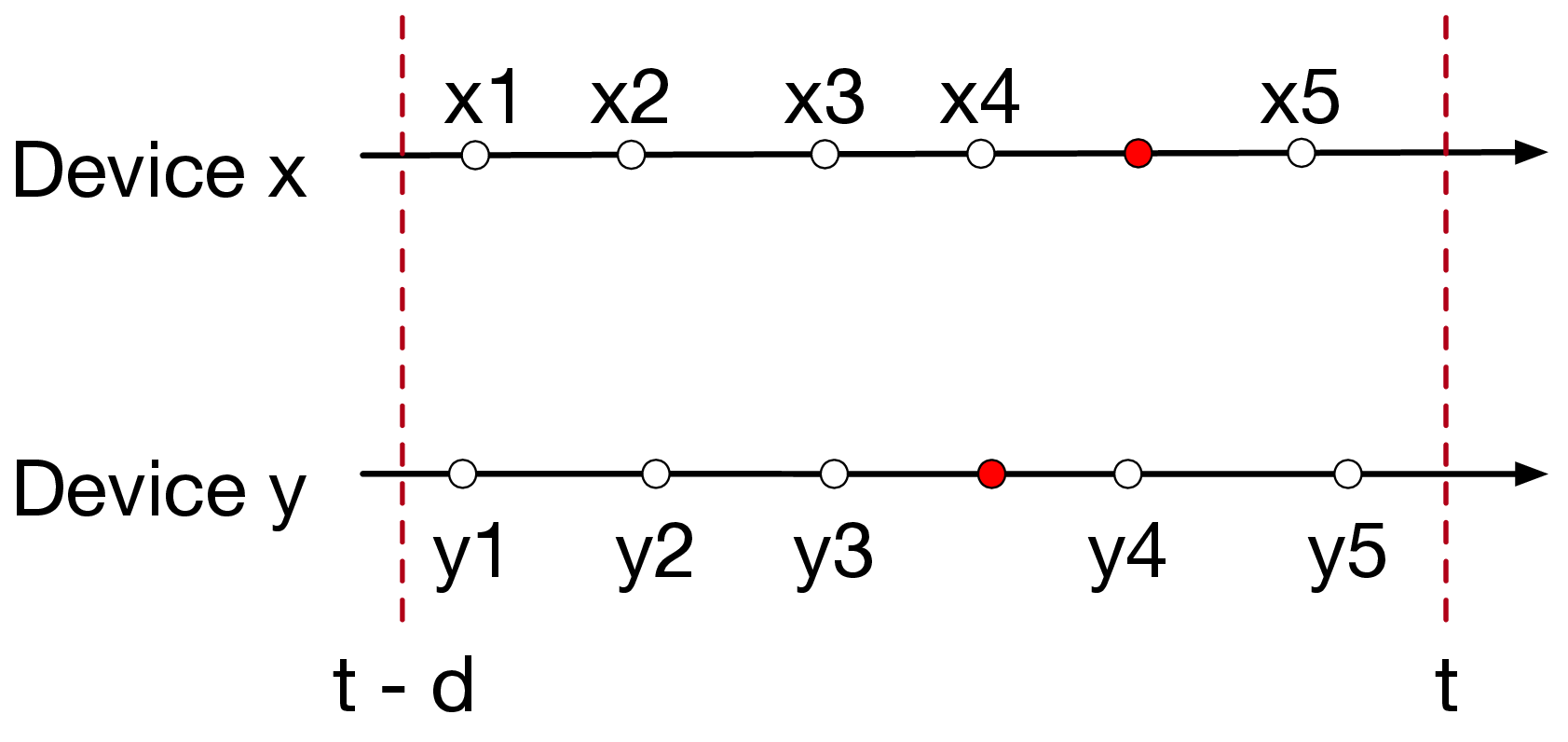}
    \label{f:align_a}
    }
    \subfigure[]{
    \includegraphics[width=.3\textwidth]{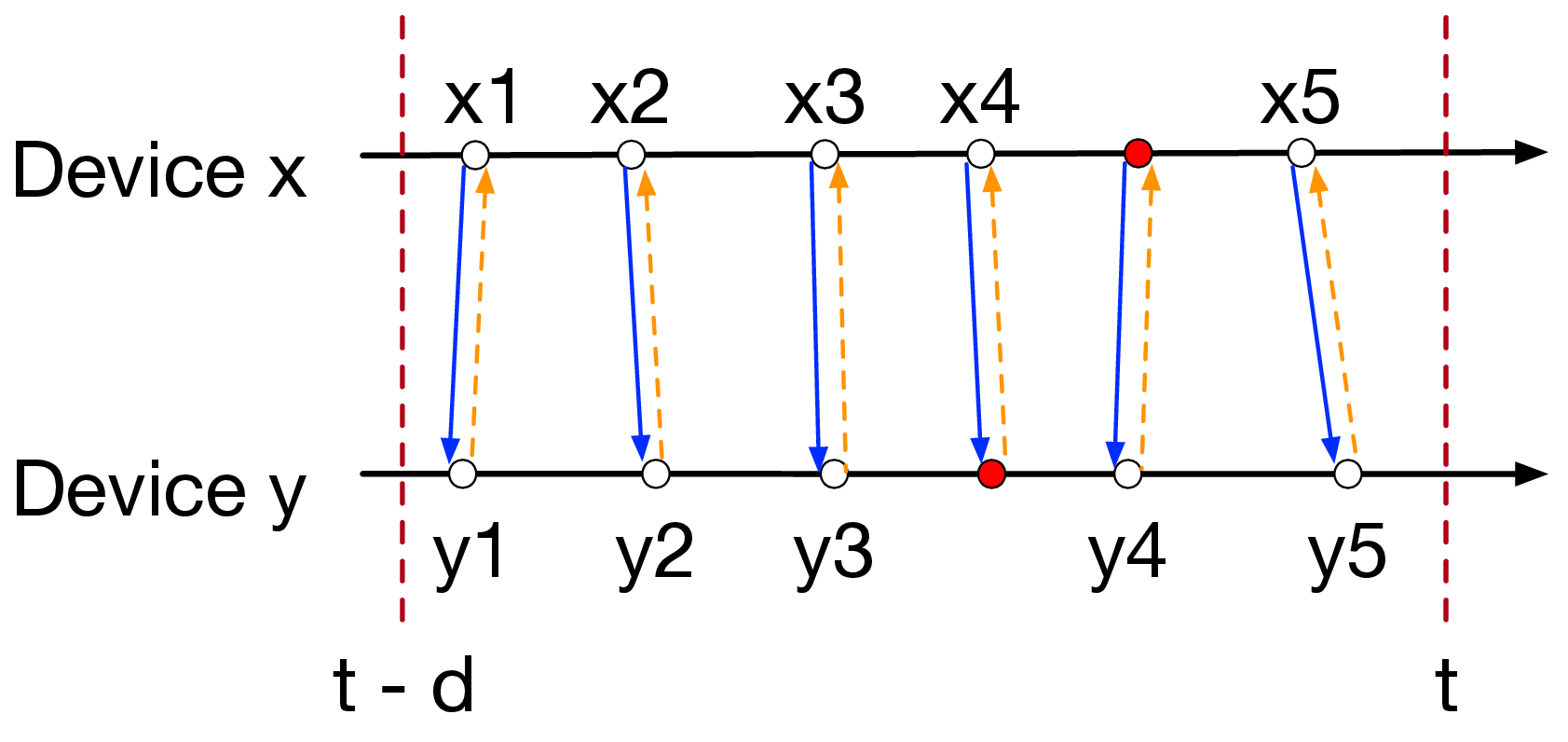}
    \label{f:align_b}
    }
    \subfigure[]{
    \includegraphics[width=.3\textwidth]{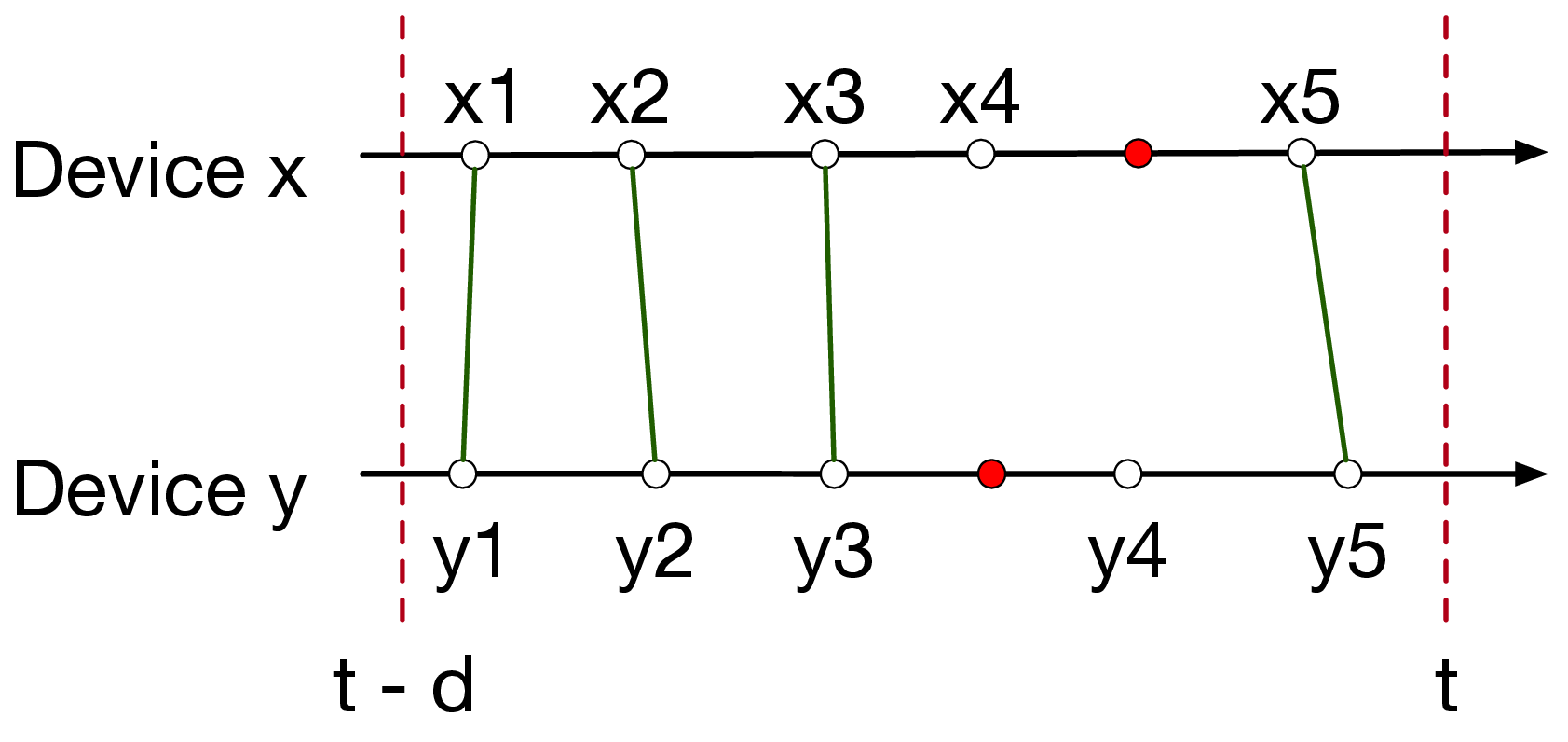}
    \label{f:align_c}
    }
    \vspace{-10pt}
    \caption{{\small{\bf This figure shows the steps in \sys's data
          pre-processing algorithm.}}\label{f:align}}
\end{figure*}
Concretely, we depict the distribution of the time intervals between a
device's two adjacent data points for all devices in our dataset in
Figure~\ref{f:missing_interval}. As we can see, the modes of the
distribution are multiples of four hours, which is the default data
collection interval of \anonisp. If the interval between two adjacent
data points is close to eight, then it is highly likely that there is
a missed data point, and so on.  The bump near 0 indicates duplicated
data points.

However, there are many data collection intervals that have lengths
between two adjacent multiples of fours, as shown in
Figure~\ref{f:missing_interval}. For example, if two data points are
six hours apart, it could either be the case of a missed data point or
the case of a delayed data collection point.  If we cannot
differentiate these two cases, we may produce a suboptimal alignment
that does not compare data points collected closest in time together.

Figure~\ref{f:align_f} shows an example. Suppose \sys runs at time $t$
and its look-back window $(t-d, t]$ includes five data points from
  cable devices $x$ and $y$, respectively. If we greedily match each
  data point in one device to its closest data point in the other
  device, we will produce an alignment as shown in
  Figure~\ref{f:align_f}, where data point $x_4$ is aligned with data
  point $y_4$. However, a better alignment, shown in
  Figure~\ref{f:align_c},  is that $x_4$ and $y_4$ are not paired
  together so the collection time difference in each paired data
  point is less. 

We design a data pre-processing algorithm to align two devices' \pnmdata points such that it minimizes the time difference between any
  two aligned data points. This algorithm involves two steps. In the
first step, we infer which data collection points are missing as shown
in Figure~\ref{f:align_a}. The red dots are the missed data points
the algorithm infers.  We note that \anonisp's \pnmdata do not include
a sequence number for each PNM record nor the timestamp when a PNM
collection request is sent.  If other cable ISPs have such
information, they can skip this step, as such information makes a
device's missed data points explicit.  In the second step, we apply a
bijection function to align the data points collected from two devices
so that each data point in one device is paired with the closest data
collection point in the other device as shown in
Figure~\ref{f:align_b}. We keep the pairs with bi-directional
alignments as the final result, as shown in Figure~\ref{f:align_c}.


To infer missed data points, we use an offline algorithm to determine
a threshold $L_{missing}$\zzy{ as another hyper-parameter of \sys}. When the data collection interval between a
device's two adjacent data points exceeds the threshold, it is highly
likely that there is a missed data point.  We use this threshold to
determine a missing data point at the time of diagnosis. We note that
this threshold is a hyper-parameter, but we choose its values
programmatically using the operational knowledge: the default data
collection interval.

As we observe in Figure~\ref{f:missing_interval}, if we project the
data collection times of all devices in the same fiber optical node on
the timeline over a long duration, we will observe
distinct clusters corresponding to each data collection epoch, where
the distances between the centers of two adjacent clusters are
multiples of the default data collection interval. If the distance
between the centers of two clusters exceeds the default data collection
interval, it indicates that there is a missed data point for the
entire \FN. Within each cluster, if a particular device
does not have a data point in the cluster, it indicates that the
device misses a data point. With this knowledge, we can determine an
optimal threshold $L_{missing}$, such that if we use the threshold to
infer a device's missed data points, the results match the earlier
cluster-based inference results the best. We show the details of
  our algorithm that determine the optimal of $L_{missing}$ in
  Appendix~\ref{appendix:missing_threshold}.

We include the pseudo-code and more details of \sys's data
pre-processing algorithm in
Appendix~\ref{appendix:infer_missing},~\ref{appendix:missing_threshold},
and~\ref{appendix:align}.  We prove in Appendix~\ref{appendix:align}
that \sys's data pre-processing algorithm produces pair-wise aligned
data points between two devices that minimize the time difference for
each pair of aligned data.



We note that the PNM infrastructure collects \pnmdata for each
  upstream channel. \sys pre-processes each channel's \pnmdata and
  then for each feature, it concatenates the data points from each
  channel into one single feature vector. For instance, the data we
  have include three upstream channels. Each feature vector is the
  concatenation of pre-processed \pnmdata from all three channels.

\section{Evaluation}
\label{sec:evaluation}

\sys's main goal is to separate maintenance issues from service issues. In this section, we evaluate how well \sys achieves this goal. 

\subsection{Establishing Evaluation Metric}
\label{sec:eval_method}

Central to \sys's design is a machine learning model that classifies a
device's state as healthy, experiencing a maintenance issue, or experiencing a
service issue.  It is challenging to evaluate its effectiveness as we
do not have the ground truth. To address this challenge and to scale 
our evaluation, we develop metrics based on customer tickets. 

Our assumption
here is that if we detect maintenance issues, from a statistical view, customers who suffer from those issues will report more maintenance tickets
compared to service tickets, and vice versa. Therefore, we define a normalized 
ticketing rate as a metric to evaluate \sys's
fault diagnosis accuracy.  Recall that in Eq~\ref{eq:rt_m} and
Eq~\ref{eq:rt_s}, we define two ticketing-rate variables
$R_{m,M}$ and $R_{m,S}$ as the number of maintenance
tickets averaged over the total length of data collection periods \sys
diagnoses as experiencing a maintenance and a service issue,
respectively. We can define ticketing rate variables $R_{s,M}$
and $R_{s, S}$ correspondingly where $s$ stands for service
tickets.

Because an ISP's fault diagnosis is inaccurate, we assume for each
type of ticket, there exists random errors. To discount those errors,
we define a baseline ticketing rate $R_{t}$ as the number of tickets
of type $t$ ($t$ is either maintenance or service tickets) averaged over all
devices in a fiber optical node and the entire data collection period.

A normalized ticketing rate $\overline{R_{t,E}}$ of a ticket
type $t$ and a diagnosis type $E$, where $E$ is either a
maintenance issue or a service issue is defined as
$\frac{R_{t,E}}{R_{t}}$. Assuming that an ISP receives 
much higher frequency of customer calls during a true maintenance or
service issue, if a fault diagnosis system can accurately
diagnose a maintenance or service issue, we will observe high 
normalized ticketing rates for both maintenance and service issues. In
contrast, if a system cannot effectively identify maintenance or
service issues, we would observe a normalized ticketing rate close to
1. 

To summarize, we define four normalized ticketing rates: 
$\overline{R_{m,M}}$, $\overline{R_{s,M}}$, 
$\overline{R_{m,S}}$, and $\overline{R_{s,S}}$. For an
algorithm that can detect service or maintenance issues from \pnmdata, we should observe the following 4 invariants:
\begin{gather*}
    \overline{R_{m,M}} > \overline{R_{s,M}} > 1, \quad
    \overline{R_{m,M}} > \overline{R_{m,S}} > 1, \\
    \overline{R_{s,S}} > \overline{R_{m,S}} > 1, \quad 
    \overline{R_{s,S}} > \overline{R_{s,M}} > 1
\end{gather*}
and high value of $\overline{R_{m,M}}$ and 
$\overline{R_{s,S}}$ indicates better performance.


\subsection{Comparing with PNM Best Practice}
\label{sec:compare_to_pnm}
The official PNM document from CableLabs~\cite{cablelabs2016docsis3}
introduces an algorithm that uses each \pnmdata point's
pre-equalization coefficients as a feature vector for clustering
devices impacted by the same linear RF distortion, which can be
caused by either a service issue or a maintenance issue. Therefore,
as a comparison, we implement this algorithm and set its clustering 
similarity threshold in the same way as we tune those thresholds for \sys's
features, and compare the normalized ticketing rate defined in \S~\ref{sec:eval_method}.

To set up the comparison experiment, We split the 14-month \pnmdata we have into two sets: an 11-month training set (from Jan 2019 to Nov 2019) and a 3-month test set (from Dec 2019 to Feb 2020). We use the training set to determine the values of \sys’s hyper-parameters \zcut{(\S~\ref{sec:parameters}) }
and the test set to evaluate \sys’s performance.

\begin{figure}[!t]
\centering
  \includegraphics[width=0.8\linewidth]{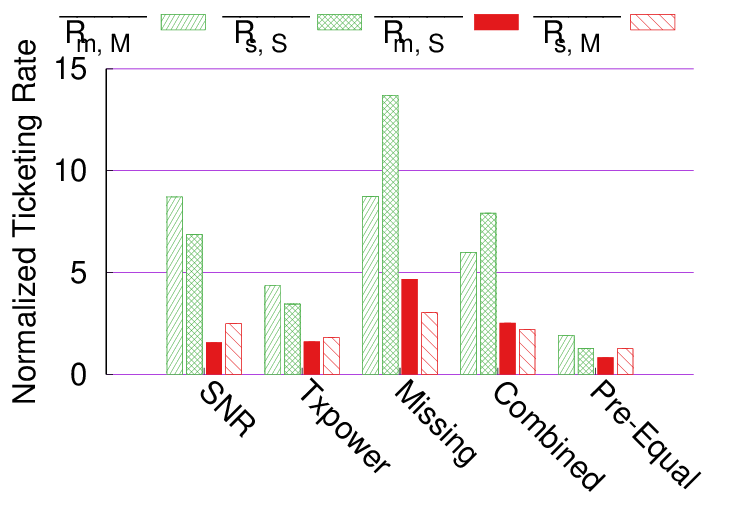}
  \vspace{-10pt}
  \caption{\small\textbf{The normalized ticketing rates for various
        combinations of 
        ticket type and diagnosis result, where $m$ denotes
        maintenance tickets, $s$: service tickets, $M$:
        \sys-diagnosed maintenance issues, and $S$:
        \sys-diagnosed service issues.}}\label{f:our_ticket_density}

\end{figure}

Figure~\ref{f:our_ticket_density} shows the results when we run \sys
on the test dataset. We show the ticketing rates for faults detected
by each independent feature as well as the ``Combined'' result
detected using all \sys's features.  For each individual metric, SNR, Tx Power and Missing, both the normalized ticketing
rates $\overline{R_{m,M}}$ and $\overline{R_{s,S}}$, which
suggests accurate diagnoses, are much higher than $\overline{R_{m,S}}$ and $\overline{R_{s,M}}$ for all \sys's features. For the missing feature, the normalized service ticketing rate during a \sys-diagnosed service issue is as high as $13$, suggesting that missing data points
are highly predictive of service issues. We also calculate the combined normalized ticketing rates as the combination of the three metrics \sys used, and it shows a good performance as well.

    In contrast, for pre-equalization coefficients, we find that all
    combinations of normalized ticketing rates are slightly above 1,
    suggesting that they are not effective in detecting or diagnosing
    customer-reported faults. Our explanation is that those
    coefficients are designed to compensate for signal distortions in
    cable networks. They 
    detect distortions that are
    already compensated for, but are not effective in signaling
    un-compensatable anomalies that lead to customer tickets.

    We also compute the normalized ticketing rates for both
    maintenance and service issues when a device is healthy. All
    values of those normalized ticketing rates are between $[0.891,
      1)$, suggesting that \sys correctly identifies healthy
      networking conditions. The normalized ticketing rates in healthy
      periods are less than 1 because the healthy periods have fewer
      than average tickets. We do not show them in
      Figure~\ref{f:our_ticket_density} for clarity.



\begin{figure}[!t]
\centering
  \includegraphics[width=0.8\linewidth]{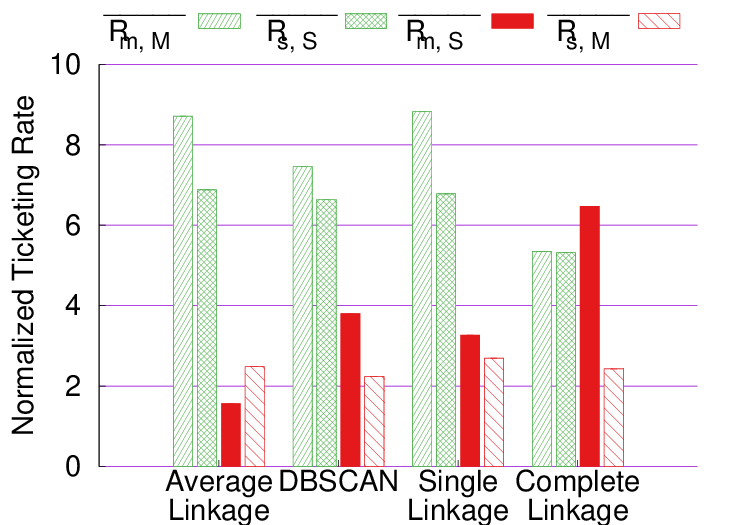}
  \caption{\small\textbf{The normalized ticketing rates of SNR for various cluster algorithms of ticket type and diagnosis result.}}\label{f:cluste_snr_density}
\end{figure}


\subsection{Comparing with Different Choices}
\label{sec:compare_choice}
In the design of \sys, we design a unique preprocessing method that includes missing position inference and time series data alignment. Meanwhile, \sys adopts the average linkage hierarchical clustering algorithm to cluster the cable devices that share the same abnormal patterns. In this subsection, we compare our design choice with other popular algorithms, to show our design choice outperforms
other algorithms. For the clustering algorithm, we compare our choice with three alternatives: DBSCAN, single-linkage, and complete-linkage clustering algorithms. For the preprocessing, we contrast our choice with resampling, a popular preprocessing method for irregularly sampled time series data. 

For each comparison, we replace our choice with alternatives and run \sys to compare the normalized ticketing rate defined in \S~\ref{sec:eval_method}. Besides, in order to validate the performance of a clustering algorithm, we manually labeled a small set of  \pnmdata
in a format similar to Figure~\ref{f:tx_power_examples} as the ground truth of this evaluation. We use
customer tickets to locate the time periods where faults are likely to
occur. If a group of devices (with group size exceeding $C_{thr}$)
shows a common anomalous pattern, we label this group of devices as
experiencing a maintenance issue. If there are only one or a few devices are affected, we label it as a service
issue. We label all devices that show no anomalous patterns as
healthy.  We learned from \anonisp that this
process resembles how they manually diagnose a network anomaly. 

We started the inspection by choosing  50 maintenance
tickets and used
  the tickets' start and close time to guide the search for anomalous
  patterns. We were able to obtain 16 groups of maintenance issues that
  impact nearly 700 devices. Since we must inspect all devices sharing
  the same fiber optical node for each maintenance ticket, we were also able to identify
  113 devices that were affected by service issues. We carefully
  verified the labeling results with \anonisp's experts to guarantee
  the labeling accuracy.
  
This manual labeling process is cumbersome and error-prone.  It
took two-person-week to obtain these labels.  We intentionally did not
expand into more labels to ensure the labeling accuracy. We
  note that this manually labeled set covers only a small fraction of
  anomalous patterns and is not suitable for training a high-quality
  classifier.
  
We run the clustering algorithm on \pnmdata we label and compare their cluster
results with our labeled results.  We choose two widely used metrics for evaluating each clustering algorithm: the Rand
Index (RI)~\cite{rand1971objective} and the Adjusted Rand Index
(ARI)~\cite{hubert1985comparing}.

We compute RI by comparing the partitions produced by a clustering
algorithm with the ground truth partition.  If two devices are in the
same cluster in both partitions, we count it as a true positive
($TP$).  Conversely, if two devices are in the same subset in the
partition produced by a clustering algorithm, but they are in
different subsets in the ground truth partition, we count it as a
false positive ($FP$). True negatives ($TN$) and false negatives
($FN$) are defined accordingly. RI computes the fraction
of true positives and negatives divided by the total pairs of devices:
$\frac{TP+TN}{TP+TN+FP+FN}$. Its maximum value is 1. The higher the
RI, the better the clustering result.  ARI adjusts for the random
chances that a clustering algorithm groups two devices in the same
cluster by deducting the expected RI ($\mathbb{E}(RI)$) of a random
partition: $\frac{RI-\mathbb{E}(RI)}{1-\mathbb{E}(RI)}$.

\paragraph{\textbf{Clustering Algorithm}}

Within the architecture of \sys, the average linkage hierarchical clustering algorithm is employed to categorize devices affected by the same network anomaly. A salient challenge is the indeterminacy of the distinct pattern count. Given this inherent uncertainty, clustering algorithms that mandate the specification of the number of clusters, represented by the hyper-parameter $k$, are inherently inconsistent with our design objectives. In the context of \sys, it is imperative to employ algorithms capable of discerning the optimal number of clusters autonomously, circumventing the limitations presented by the need for predefined cluster counts. Therefore, we compare \sys's clustering algorithm choice with three popular clustering algorithms that are not contingent on the predefined $k$ value, including DBSCAN, single-linkage, and complete-linkage clustering algorithms.

\begin{table}[b!]
\scriptsize
\centering
\begin{tabular}{c c c c c}
                    & \makecell[c]{Average \\ Linkage} & DBSCAN & \makecell[c]{Single \\ Linkage}  & \makecell[c]{Complete \\ Linkage} \\ \toprule
RI                  & 0.91  & 0.84 & 0.83 & 0.83 \\ \midrule
ARI                 & 0.83  & 0.65 & 0.64 & 0.66 \\ 
\end{tabular}
\caption{\small\textbf{Rand Index and Adjusted Rand Index for various clustering algorithms.}\label{tab:manual_label_result}}
\end{table}

Table~\ref{tab:manual_label_result} shows the comparison results. \sys
achieves an RI of 0.91 and an ARI of 0.83, respectively. \sys's choice outperforms other clustering algorithms. 

Figure~\ref{f:cluste_snr_density} 
shows the normalized ticketing rate for SNR\zzy{ (the result of Tx Power is similar but skipped due to space limitation)}. The figure shows that the average linkage hierarchical clustering algorithm achieves the highest $\overline{R_{m,M}}$ and $\overline{R_{s,S}}$, which suggests the best performance. It is pertinent to note that the comparative analysis does not encompass the normalized ticketing rate for missing data. This exclusion is attributed to the observation that various clustering algorithms exhibit analogous performance metrics in this dimension. The uniformity in the normalized ticketing rate for missing data across different algorithms is caused by the characteristic that the missing vectors are not subjected to irregular sampling, and the missing vector is a 0-1 vector, rendering the comparative distinctions negligible.

\paragraph{\textbf{Preprocessing}}

\begin{figure}[!t]
\centering
  \includegraphics[width=0.8\linewidth]{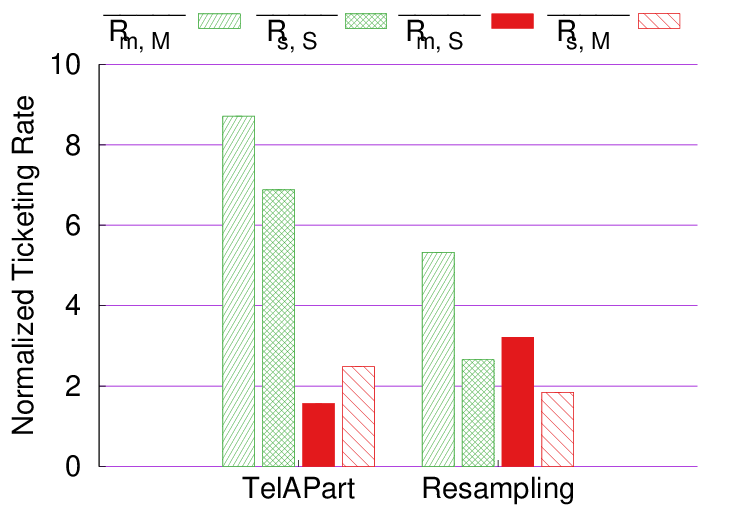}
  \vspace{-10pt}
  \caption{\small\textbf{The normalized ticketing rates of SNR for two preprocessing algorithms of ticket type and diagnosis result.}}\label{f:cluste_snr_density_preprocessing}
\end{figure}


A common method to handle irregularly sampled time series data involves its transformation into uniformly spaced intervals through resampling~\cite{steven2020learning}. Each time series is systematically restructured onto a consistent grid, necessitating the application of interpolation techniques or, in instances of noise prevalence, regression models to approximate the inherent continuous temporal dynamics.

We employ a classic resampling technique to process the \pnmdata, utilizing linear interpolation to transform the data into a uniformly spaced 4-hour interval. This method effectively aligns the data points and mitigates the presence of missing data, resulting in a coherent and complete dataset. We then run the clustering algorithm \sys adopts to demonstrate the effectiveness of our preprocessing algorithm.

\begin{table}[b!]
\scriptsize
\centering
\begin{tabular}{c c c}
                    & \makecell[c]{\sys's\\ Preprocessing} & Resampling \\ \toprule
RI                  & 0.91  & 0.74  \\ \midrule
ARI                 & 0.83  & 0.35  \\ 
\end{tabular}
\caption{\small\textbf{Rand Index and Adjusted Rand Index for two preprocessing methods}\label{tab:preprocessing_compare}}
\end{table}

Table~\ref{tab:preprocessing_compare} shows the comparison results. \sys achieves an RI of 0.91 and an ARI of 0.83, respectively. \sys's preprocessing algorithm choice outperforms the classic resampling algorithm.  

Figure~\ref{f:cluste_snr_density_preprocessing} 
shows the normalized ticketing rate for SNR\zzy{ (the result of Tx Power is similar but skipped due to space limitation)}. The figure reveals that the preprocessing algorithm employed within \sys outperforms alternative methods, as evidenced by the elevated values of $\overline{R_{m,M}}$ and $\overline{R_{s,S}}$.

\begin{figure*}[!t]
\minipage{0.32\textwidth}
  \includegraphics[width=\linewidth]{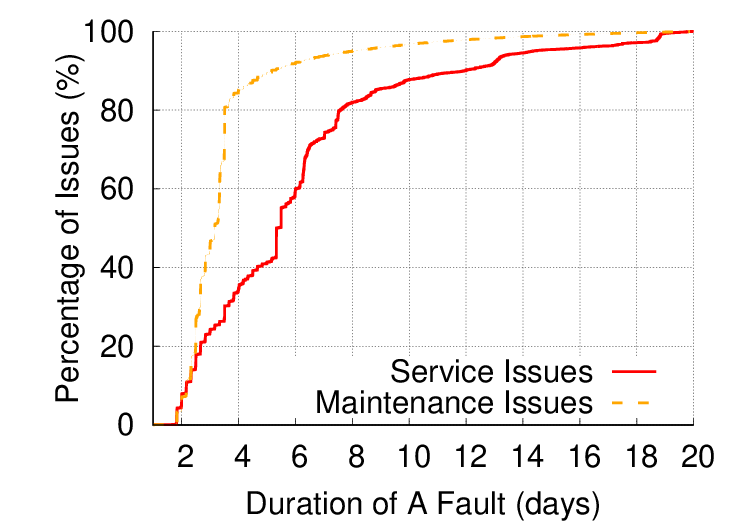}
  \caption{\small\textbf{The cumulative distribution of how long a
        maintenance (or service) issue lasts.}}\label{f:network_issue_existing_time}
\endminipage\hfill
\minipage{0.32\textwidth}
  \includegraphics[width=\linewidth]{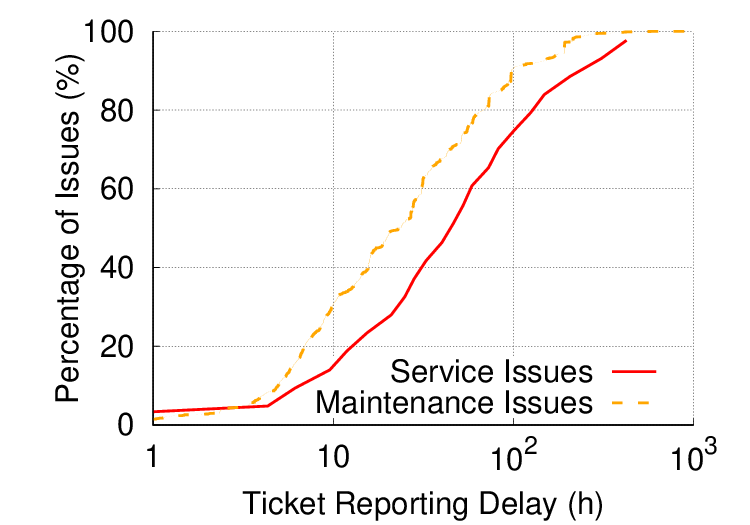}
  \caption{\small\textbf{The cumulative distributions of ticket
        reporting delays for \sys-diagnosed maintenance and service issues.}}\label{f:report_waiting_time}
\endminipage\hfill
\minipage{0.32\textwidth}
  \includegraphics[width=\linewidth]{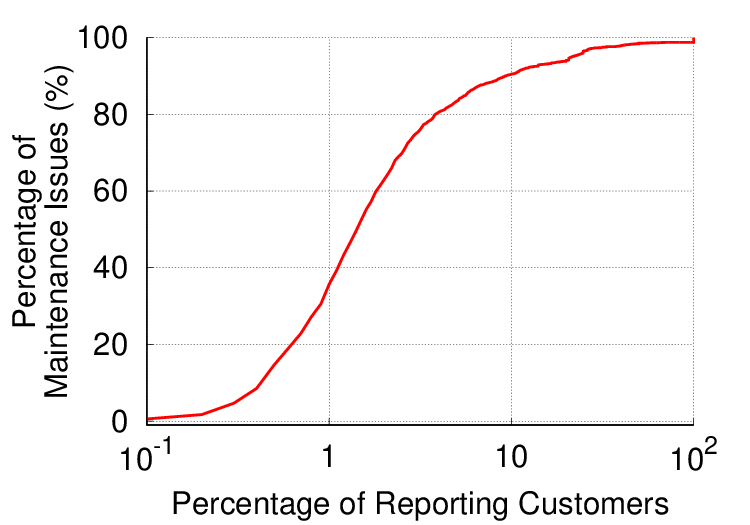}
  \caption{\small\textbf{The cumulative distributions of reporting customers
        when a maintenance issue occurs.}}\label{f:ticket_report_rate}
\endminipage\hfill
\end{figure*}

\subsection{Fault Characteristics}

An additional source of truth we have is how \anonisp processes
customer calls and the definitions of maintenance and service issues.
Due to the lack of any other source of truth, we use this operational
knowledge to validate \sys's diagnosis accuracy.

\paragraph{\textbf{Fault Duration}} By \anonisp's
definition of the maintenance cluster size threshold, a maintenance
issue impacts more customers than a service issue. Therefore, more
customers are likely to call an ISP when a maintenance issue occurs
than when a service issue occurs.  If \sys correctly differentiates a
maintenance issue from a service issue, we would observe that the
delay between when a maintenance issue occurs and when the first
customer calls to be shorter than that when a service issue occurs. As
a result, a maintenance issue is likely to be fixed sooner than a
service issue.


To test this hypothesis, we compute the duration of a \sys-diagnosed
fault as follows. We run \sys in \proactive on the test dataset and
record the first time a fault is detected until the time the fault is
no longer detected.  In Figure~\ref{f:network_issue_existing_time}, we
depict the cumulative distribution of the fault duration for
maintenance and service issues, respectively. As can be seen, 80\% of
the maintenance issues are fixed in 84 hours (3.5 days), while 80\% of the
service issues are fixed in 182 hours (7.6 days). This result further indicates
\sys's diagnosis is accurate.



\paragraph{\textbf{Ticket Reporting Delay}}

Figure~\ref{f:report_waiting_time} shows the distributions of ticket
reporting delays for maintenance issues and service issues diagnosed
by \sys, respectively.  We do not differentiate the first arrival
ticket's type for a \sys-diagnosed issue, as it may be mis-labeled by
an operator. We observe that when the ticket reporting delay exceeds a
few hours, there is a significant difference between the ticket
reporting delay distribution for maintenance issues and that for
service issues. For example, for more than 50\% of the maintenance
issues diagnosed by \sys, the first ticket arrives within 24 hours. In
contrast, only for less than 31\% of the service issues diagnosed by
\sys, the first ticket arrives within 24 hours. This difference again
suggests that \sys is able to separate maintenance issues from service
issues effectively. Interestingly, when the ticket
reporting delay is short, the ticket reporting delay distributions for
maintenance and service issues overlap. We hypothesize that these
tickets are impacted by severe issues  so that any subscriber impacted
by one of these issues reports immediately. 

\subsection{User Ticketing Behavior}

We are interested in studying what percentage of customers will make a
trouble call when a customer-impacting maintenance issue occurs. This
study does not serve the purpose of evaluation, but offers useful
information to researchers and network operators. For this purpose, we
examine each maintenance issue diagnosed by \sys during the span of the
entire dataset and measure the size of the anomalous cluster \sys
detects. We then count the number of customer tickets coinciding with
the duration of the maintenance issue and divide it by the size of the cluster.

Figure~\ref{f:ticket_report_rate} shows the cumulative distribution of
the percentage of reporting customers for a maintenance issue. For
more than 90\% of the maintenance issues, 90\% of the impacted
customers will not make a call. The low percentage of
reporting customers suggests that it is important for ISPs to
periodically monitor their networks and proactively repair network
impairments to improve the quality of experience for their
customers. ISPs can invoke \sys in the \proactive periodically to
achieve this goal.

\zcut{\paragraph{Summary} We evaluate \sys's diagnosis accuracy using a
small set of manually labeled dataset, normalized ticketing rates, and
operational knowledge. All evaluation methods suggest that \sys can
accurately identify maintenance issues and service issues. This
promising result has prompted \anonisp to incorporate \sys into their
\pnmdata pipeline.}

\zcut{
\subsection{Field Test}

To demonstrate the effectiveness of \sys, we collaborate with \anonisp to perform a field test. \sys is wrapped as a service deployed with Docker that can be called by \anonisp's field team. We elaborately designed the API of the \sys service such that the service can be incorporated into \anonisp's workflow. 
We received a summary but no details from the ISP, confirming \sys's effectiveness: ``(we) evaluated the performance of the clustering methodology produced by (your) team and found it effective at the task of classifying defects as service or maintenance''\footnote{Sensitive names are hidden for anonymity.}.
}

\section{Conclusion}
\label{sec:conclusion}

Cable ISPs can benefit from accurate and automated fault diagnosis for
reducing operational costs caused by erroneous fault diagnoses.  This
work addresses a fault diagnosis problem present in cable broadband
networks, \ie, how to distinguish a maintenance issue from a service
issue.  We develop \sys, a system that uses the telemetry data readily
available in cable broadband networks to automatically diagnose the
type of fault. In \sys's design, we combine unsupervised learning,
anomaly detection, and optimization techniques to eliminate the need
for manually labeled training data and hand-tuning
hyper-parameters. We also develop data pre-processing techniques to
enable machine learning models on \pnmdata, which are unreliably
collected and contain missing, duplicated, and unaligned data points.
We use a small set of manually labeled data and customer ticket
statistics to evaluate \sys. The evaluation results show that when
compared to the labeled data, \sys achieves a Rand Index of
0.91. During \sys-diagnosed maintenance (or service) events, a much
higher than average frequency of maintenance (or service) tickets
occur, further suggesting that \sys can effectively distinguish
maintenance issues from service issues. The cable ISP we collaborated with has confirmed the effectiveness of \sys with field tests.





\bibliographystyle{IEEEtran}
\bibliography{paper_ton}



\appendix

\section{Algorithm for Inferring Missing Data Points}
\label{appendix:infer_missing}

\begin{algorithm}[H]
\caption{Inferring Missing Data Points}
\label{alg:inter_missing}
\begin{algorithmic}[1]
\State $T_{a\_raw} = \{ t_1, t_2, ..., t_n \}$ \Comment{The timestamp array of the time series data for device $a$} 
\State $T_a \leftarrow \{t_1\}$ \Comment{The array to store both the valid timestamp and missing timestamp}
\For {$i \ $ from $2 \ $ to $n \ $}
\While {$t_i$ - $t_{i-1}$ in $T_a$ >= $L_{missing}$}
\State insert a placeholder at  $t_{i-1} + 4hrs $ in $T_a$
\EndWhile
\State insert $t_{i}$ in $T_a$
\EndFor
\State return $T_a$
\end{algorithmic}
\end{algorithm}
This algorithm will iterate each data point once. Suppose the input data length is K, and the total number of devices is M. Then the time complexity of this algorithm for each device is O(K), and running this algorithm among all the devices will cost $O(M \times K)$.

\section{Determine the Optimal Missing Threshold}
\label{appendix:missing_threshold}
The algorithm for inferring missing data points required an accurate missing threshold $L_{missing}$. 
It is challenging to find the best $L_{missing}$ because
if we set this value too aggressively, we may infer more missing data points than ground truth. And if we set this value too conservatively, we may fail to infer some missing data points. 

To overcome this challenge, we developed an algorithm to determine the optimal missing threshold. The high-level idea is, using a given $L_{missing}$, we can infer the missing data points in our \pnmdata. We then compare our inferred results with the ground truth we generated, and select the $L_{missing}$ that gives us the highest accuracy. In \sys, all the steps are automatic to cable ISPs once they input the \pnmdata and the length of the default data collection interval $L$. We will first introduce the method we use to generate the ground truth of missing data points, then describe how we use this ground truth to guide us to find the best $L_{missing}$. 

\paragraph{Generate the Ground Truth of Missing Data Points}
\begin{figure}[h]
\centering
    \includegraphics[width=.49\textwidth]{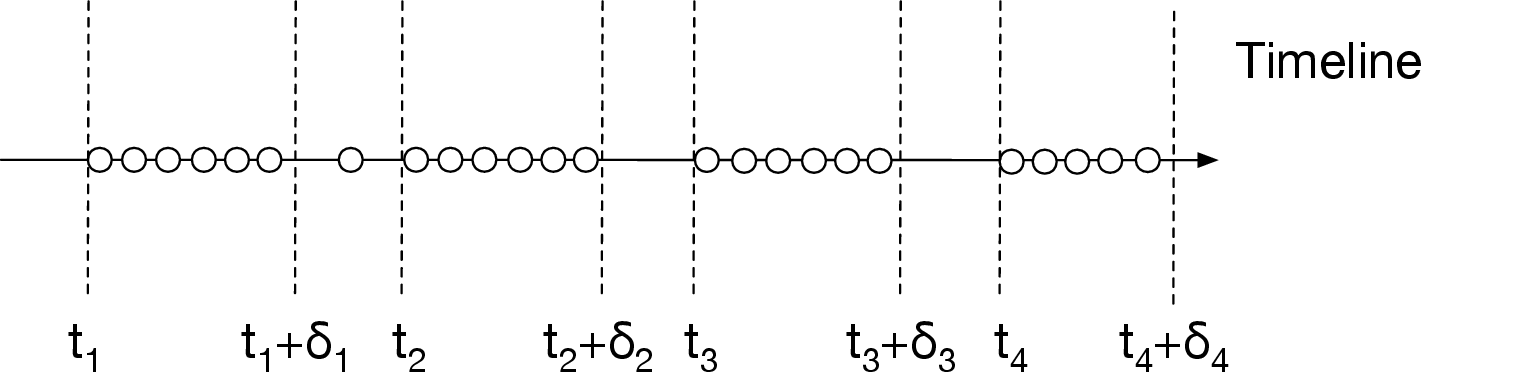}
    \caption{\small\textbf{This figure shows the data collection time for all devices in the same \FN.}}
    \label{f:data_collection_example}
\end{figure}
Figure~\ref{f:data_collection_example} shows an example of the data collection time for all devices in the same \FN in our data. \anonisp starts to collect the \pnmdata from this \FN at time $t_1$, $t_2$, $t_3$, and $t_4$, and finishes each round of data collection at time $t_1+\delta_1$, $t_2+\delta_2$, $t_3+\delta_3$, and $t_4+\delta_4$, respectively. The time interval between $t_1$, $t_2$, $t_3$, and $t_4$ are roughly the same as the default data collection interval $L$ (4 hours in this example). At each data collection period, \pnmdata will be collected from all devices in random order. Therefore, it is possible that the \pnmdata from a device was collected at $t_1$ in the first period, and was collected at $t_2+\delta_2$ in the second period. This observation provides evidence of why inferring missing data points and determining the optimal missing threshold is a challenge.

Figure~\ref{f:data_collection_example} shows the data points have a high density in each data collection period, while between two periods, the data point distribution is very sparse. This inspires us to adopt the density-based clustering algorithm to find each data collection interval. Once we observed the \pnmdata from a device are only collected at the first and the third data collection periods, then we can infer that in the second data collection period, the \pnmdata from that device are missing. Once we obtain all the data collection periods, we can use this to help us find all the missing data points in our \pnmdata. We use this result as the ground truth of missing data points.

We use DBSCAN, one of the most popular density-based clustering algorithms to help us find the data collection periods. DBSCAN requires two hyper-parameters: epsilon and min\_samples. We use grid search to find the optimum values of these two hyper-parameters. For a given epsilon and min\_sample values, we run DBSCAN to generate the data collection periods $P = \{(t_i, t_i+\delta_i)|i\in[1, n]\}$. We define the estimated error:
$$e=\sum_{i=2}^{n}(L-(t_i+\frac{\delta_i}{2}) + (t_{i-1}+\frac{\delta_{i-1}}{2}))$$
We select the epsilon and min\_sample values that minimize the estimated error and use the DBSCAN with such hyperparameters to generate the ground truth of missing data points in \pnmdata.

\paragraph{Determine the Optimal Missing Threshold}
We use the ground truth of missing data points we obtained to help us determine the optimal missing threshold $L_{missing}$. For a given $L_{missing}$ value, our infer missing algorithm (Algorithm~\ref{alg:inter_missing}) can generate an inferring result.  For two data points $x_i$ and $x_{i+1}$ collected from a device, assume the ground truth shows there are $n$ missing data points between $x_i$ and $x_{i+1}$, and Algorithm~\ref{alg:inter_missing}
reports there are $m$ missing data points. If $n=m$ and $n \neq 0$, we label those $n$ missing data points as true positive ($TP$). If $n=m=0$, we label it as a true negative ($TN$). If $n>m$, we will mark there are $m$ true positives and $n-m$ false negatives ($FN$). If $n<m$, we will count there are $n$ true positives and $m-n$ false positives ($FP$). By comparing our inferring missing results and the ground truth, we can calculate our infer accuracy $acc=\frac{TP+TN}{TP+TN+FP+FN}$. We use grid search to help us find the $L_{missing}$ that gives us the highest accuracy and use this $L_{missing}$ to infer missing data points in \pnmdata in this work.

\section{Alignment Algorithm}
\label{appendix:align}

\begin{algorithm}[H]
\caption{Pairwise Alignment}
\label{alg:align}
\begin{algorithmic}[1]
\State $T_x = \{ t_{x1}, t_{x2}, ..., t_{xn} \}$ \Comment{The timestamp array of the time series data for device $x$ including missing timestamps} 
\State $T_y = \{ t_{y1}, t_{y2}, ..., t_{ym} \}$ \Comment{The timestamp array of the time series data for device $y$ including missing timestamps} 
\State $S_{(x, y)} \leftarrow \{(t_{xi}, t_{yj})\ | \ i \in [1, n]\} \ $   where $\ t_{yj} \ = \argmin |t_{xi}-t_{yk}|, k \in [1, m]$
\State $S_{(y, x)} \leftarrow \{(t_{yi}, t_{xj})\ | \ i \in [1, m]\} \ $   where $\ t_{xj} \ = \argmin |t_{yi}-t_{xk}|, k \in [1, n]$
\State $A_{(x, y)} \leftarrow \{(t_{xi}, t_{yj})\ | \ (t_{xi}, t_{yj}) \in S_{(x, y)} \ and \ (t_{yj}, t_{xj}) \in S_{(y, x)}, \ i \in [1, n], \ j \in [1, m] \} $
\State return $A_{(x, y)}$
\end{algorithmic}
\end{algorithm}
Suppose the input data length is K, and the total number of devices is M. In Algorithm~\ref{alg:align}, for each pair-wise devices, it will iterate each device's time series data twice. So the time complexity for each pair-wise devices is O(K). Since \sys requires the alignment for all pair-wise devices, the overall time complexity will be $O(M^2 \times K)$.

Algorithm~\ref{alg:align} shows how we align the data points. Intuitively, the algorithm finds a bijection of the nearest data points for alignment. And it ensures that for any timestamp pair of any alignments with the number of timestamp pairs equal to or more than the result alignment, the result alignment must have an aligned timestamp pair with shorter time skew. Formally speaking:

\begin{theorem}
\label{theorem:align}
Given two time series data $T_x = \{ t_{x1}, t_{x2}, ..., t_{xn} \}$ and $T_y = \{ t_{y1}, t_{y2}, ..., t_{ym} \}$, Algorithm~\ref{alg:align} returns an alignment $A_{(x, y)} = \{(t_{xi}, t_{yj})\}$. Then $\forall\tilde{A}_{(x, y)} = \{(t_{\tilde{xi}}, t_{\tilde{yj}})\}$ such that $|\tilde{A}_{(x, y)}| \geq |A_{(x, y)}|$, we have $\forall (t_{\tilde{xi}}, t_{\tilde{yj}})\in\tilde{A}_{(x, y)}\exists (t_{xi}, t_{yj})\in A_{(x, y)}(|t_{xi} - t_{yj}| \leq |t_{\tilde{xi}} - t_{\tilde{yj}}|)$.
\end{theorem}

\begin{proof}

For any $(t_{\tilde{xi}}, t_{\tilde{yj}})\in\tilde{A}_{(x, y)}$, we discuss 3 scenarios regarding to if any timestamp in the pair is covered by a timestamp pair in $A_{(x, y)}$.

\begin{enumerate}
    \item $\exists (t_{xi}, t_{yj})\in A_{(x, y)}(t_{xi} = t_{\tilde{xi}})$.
    Because Algorithm~\ref{alg:align} always chooses the closest timestamp when pairing them, $|t_{xi} - t_{yj}| = |t_{\tilde{xi}} - t_{yj}| \leq |t_{\tilde{xi}} - t_{\tilde{yj}}|$.
    \item $\exists (t_{xi}, t_{yj})\in A_{(x, y)}(t_{yj} = t_{\tilde{yj}})$. Similarly, $|t_{xi} - t_{yj}| = |t_{xi} - t_{\tilde{yj}}| \leq |t_{\tilde{xi}} - t_{\tilde{yj}}|$.
    \item Otherwise, both $t_{\tilde{xi}}$ and $t_{\tilde{yj}}$ do not exist in any pairs of $A_{(x, y)}$. Thus, they cannot be the closest timestamp to each other. Without loss of generality, we assume $\exists \tilde{yj_2}(|t_{\tilde{xi}} - t_{\tilde{yj}}| > |t_{\tilde{xi}} - t_{\tilde{yj_2}}|)$ and $!\exists y(|t_{\tilde{xi}} - t_{\tilde{yj_2}}| > |t_{\tilde{xi}} - t_{y}|)$. Then if $\exists (t_{xi}, t_{yj})\in A_{(x, y)}(t_{yj} = t_{\tilde{yj_2}})$, we fall back to the first two cases and have $|t_{\tilde{xi}} - t_{\tilde{yj}}| \geq |t_{\tilde{xi}} - t_{\tilde{yj_2}}| \geq |t_{xi} - t_{yj}|$, where $(t_{xi}, t_{yj})\in A_{(x, y)}$ is the timestamp pair we found in the first two cases. Otherwise, we can repeat this process. Because both $T_x$ and $T_y$ are finite set, this process will terminate and we have $|t_{\tilde{xi}} - t_{\tilde{yj}}| \geq |t_{\tilde{xi}} - t_{\tilde{yj_2}}| \geq |t_{\tilde{xi}} - t_{\tilde{yj_3}}| \geq ... \geq |t_{xi} - t_{yj}|$.
\end{enumerate}

\end{proof}

This result shows that Algorithm~\ref{alg:align} is selecting timestamps as close as possible when pairing them. And thus, with high probability, Algorithm~\ref{alg:align} can align the points collected in the same collection period even with time skew.

\section{Ticket Guided Clustering}
\label{sec:cluster_proof}

Let $R_{m,M}$ denote the maintenance ticketing rate in detected maintenance issues, and $R_{m,S}$ denote the maintenance ticketing rate in detected service issues, our target metric is defined as $TRR_m = \frac{R_{m,M}}{R_{m,S}}$. We now prove that maximizing $TRR_m$ could provide us with the best clustering hyper-parameter $s_f$ defined in \S~\ref{sec:clustering}, which could minimize false positives and false negatives.

For a clustering algorithm, we start with a set of devices $\mathbf{D} = \mathbf{M} \cup \mathbf{S} \cup \mathbf{U}$, where $\mathbf{M} = \{m_1, m_2, ...\}$ is the set of devices with actual maintenance issues (based on the unknown but existing ground truth), $\mathbf{S} = \{s_1, s_2, ...\}$ is the set of devices with actual service issues and $\mathbf{U} = \{u_1, u_2, ...\}$ is the set of devices without any issues. Our proof is based on the following assumptions: (1) the similarity among devices with actual maintenance issues is always the highest, compared to the similarity between any 2 devices when they do not have maintenance issues simultaneously. This assumption is made based on the observation in \S~\ref{sec:challenges}. And as a result, we would expect devices in $\mathbf{M}$ will be merged first during the hierarchical clustering process. (2) Maintenance tickets will only be reported by devices in $\mathbf{M}$, and has a uniform distribution on those devices (with the reporting ticket number expectation as $p > 0$).

\begin{theorem}
Given a set of devices $\mathbf{D} = \mathbf{M} \cup \mathbf{S} \cup \mathbf{U}$, maximizing $TRR_m$ could provide us the best clustering hyper-parameter $s_f$.
\end{theorem}

\begin{proof}

To begin with the proof, we first clarify the criteria for determining how good a clustering result is: an ideal result should maximize the Rand Index (RI) and the Adjusted Rand Index (ARI) (as defined in \S~\ref{sec:compare_choice} for evaluation). The higher those 2 metrics are, the better the clustering result is. 

In order to prove this result, we firstly proof that the best clustering hyper-parameter $s_f$ we choose can minimize the false positive ($FP$), \ie, a device without any maintenance issues is detected as with a maintenance issue, and the false negative ($FN$), \ie, a device with a maintenance issue is detected as without any maintenance issues. Then we show its equivalence to maximize RI.

We divide the entire hierarchical clustering process into 2 stages divided by the time point that all devices in $\mathbf{M}$ have been merged into maintenance clusters and then devices from $\mathbf{S}$ and $\mathbf{U}$ begin to participate into the merging, which is also the start of the second stage.

At the first stage, every time we merge a new device in $\mathbf{M}$ into a maintenance cluster would reduce $FN$. We will not incorrectly classify devices in $\mathbf{S}$ or $\mathbf{U}$ as with a maintenance issue yet and therefore $FP$ would remain stable. On the other hand, $R_{m,M}$ will remain stable because of the uniform maintenance ticket distribution, but $R_{m,S}$ decreases because we correctly identify more devices in $\mathbf{M}$. As a result, $TRR_m$ would increase at this stage. Conclusively, at this stage, we have:
$$\frac{\partial\,TRR_m}{\partial FN} < 0,~ \frac{\partial\,TRR_m}{\partial FP} = 0$$

At the second stage, we will not incorrectly classify any devices in $\mathbf{M}$ anymore and therefore $FN$ would remain stable. However, every time we merge a new device in $\mathbf{S}$ or $\mathbf{U}$ into a maintenance cluster will increase $FP$. On the other hand, $R_{m,S}$ would remain stable because all devices in $\mathbf{M}$ have already been merged into maintenance clusters. However, when we merge a new device in $\mathbf{S}$ or $\mathbf{U}$ into a maintenance cluster, because this new device will not report any maintenance tickets and $p > 0$, $R_M$, as the averaged ticketing rate, will decrease. As a result, $TRR_m$ would decrease at this stage. Conclusively, at this stage, we have:
$$\frac{\partial~TRR_m}{\partial FN} = 0,~ \frac{\partial~TRR_m}{\partial FP} < 0$$

Combining these 2 stages, we have:
$$\frac{\partial\,TRR_m}{\partial FN} \leq 0,~ \frac{\partial\,TRR_m}{\partial FP} \leq 0$$

which indicates that when $TRR_m$ is maximized, the clustering result of the corresponding $s_f$ minimizes both $FP$ and $FN$.

Next, RI is defined with true positive ($TP$) and true negative ($TN$) together with $FN$ and $FP$:

$$RI = \frac{TP+TN}{TP+TN+FP+FN} = 1 - \frac{FP+FN}{TP+TN+FP+FN}$$

Because $TP+TN+FP+FN$ is the total number of devices and is a constant, we can conclude that RI is maximized when $FP+FN$ is minimized, which in turn when $TRR_m$ is maximized as how we choose the best $s_f$. ARI essentially is a normalized version of RI ($ARI = \frac{RI-\mathbb{E}(RI)}{1-\mathbb{E}(RI)}$) and shares the same characteristics and thus is also maximized with the best $s_f$.

\end{proof}

\newpage

 


\vspace{11pt}


\vfill

\end{document}